\documentclass{article}
\usepackage{arxiv}

\usepackage[utf8]{inputenc} % allow utf-8 input
\usepackage[T1]{fontenc}    % use 8-bit T1 fonts
\usepackage{hyperref}       % hyperlinks
\usepackage{url}            % simple URL typesetting
\usepackage{amsfonts}       % blackboard math symbols
\usepackage{graphicx}
\usepackage{amsthm}
\usepackage{amsmath}
\usepackage{amssymb}
\usepackage{natbib}
\usepackage{caption}
\usepackage{subcaption}
\usepackage{enumerate}
\usepackage{color}
\usepackage{multirow}
\usepackage{placeins}
\usepackage{longtable}
\usepackage{algpseudocode}

\usepackage[ruled]{algorithm2e}
\SetKwInput{KwInput}{Input}     
\SetKwInput{KwOutput}{Output}
\SetKwBlock{stepone}{Step 1:}{}
\SetKwBlock{stepgrid}{Step 2: Grid Traversal}{}
\SetKwBlock{stepbisec}{Step 3: Bisection Search}{}
\SetKwBlock{intersectionSearch}{Intersection Search}{}
\SetKwBlock{init}{Initialization:}{}
\SetKwBlock{iter}{Iteration:}{}
\usepackage{setspace}

\usepackage{titlesec}
\usepackage[dvipsnames]{xcolor}

\usepackage{tikz}
\usetikzlibrary{patterns}
\usetikzlibrary{arrows}
\usetikzlibrary{decorations,shapes,shadings}
\usetikzlibrary{patterns}
\usetikzlibrary{decorations.pathreplacing}
\usetikzlibrary{3d,calc}

\newtheorem{Def}{Definition}
\newtheorem{lemma}{Lemma}
\newtheorem{theorem}{Theorem}

\renewcommand{\subset}{\subseteq}

\title{Principles in harmony:
Closed testing meets the partitioning principle for computational efficiency}
\author{
 Werner Brannath\thanks{Corresponding author: Werner Brannath} \\
  Institute for Statistics and\\
  Competence Center for  Clinical Trials Bremen\\
  University of Bremen \\
  \texttt{brannath@uni-bremen.de}\\
  \And
   Liane Kluge\\
  Competence Center for Clinical Trials Bremen\\
  University of Bremen \\
  \texttt{liane@uni-bremen.de}\\
}

\begin{document}
\maketitle
\begin{abstract}
We explore and utilize the algorithmic relationship between the closed testing principle for multiple tests with family-wise error rate (FWER) control and the partitioning principle for the construction of simultaneous confidence intervals. Starting with the simple observation that a multiple test with FWER control is formally equivalent to a one-sided simultaneous confidence interval for the vector of binary parameter indicating whether the null or alternative hypothesis is true, we show that the closed testing and partitioning principles follow the same computational approach. We will then utilise this relationship to extend concepts of consonance for closed tests to the partitioning principle, with the aim of deriving computationally feasible and efficient algorithms for the calculation of simultaneous confidence intervals. We will also utilize the relationship between closed testing and partitioning principle to extend common closed testing procedures to simultaneous confidence intervals, referencing the existing literature on informative simultaneous confidence intervals. The relationships and extensions will be illustrated by simple, instructive examples.
\end{abstract}

% keywords can be removed

\section{Introduction}

The closed testing and partitioning principles are fundamental principles of statistical multiple inference. While the now 50 years old closed testing principle suggested by \citet{marcus1976closed} (see also the paper of \citealp{sonnemann1982}; and its english translation by \citealp{finner2008}) provides a general framework for the construction of multiple tests with family wise error rate control (FWER), the partitioning principle has been introduced by \citet{stefansson1988confidence} for the construction of simultaneous confidence intervals, in particular for those that are consistent with step-wise multiple tests. \citet{finner2002partitioning} formalized and generalized the partitioning principle and investigated its utility for the construction of multiple tests as well as selection and ranking procedures. 

In this paper we focus on the construction and efficient calculation of simultaneous confidence intervals (SCI) that are e.g.\ compatible or inspired by closed testing procedures. We start with the observation that any multiple test with control of the FWER can be viewed as a one-sided SCI for the multiple binary parameter that indicate which null or alternative hypotheses are true. Based on this observation, we can understand the partitioning principle for SCI as an extension of the closed testing principle. This then leads us to explore the question of how far the concept of ``consonance'' for closed testing procedures (see e.g.\ \citealp{hommel2007powerful}; \citealp{brannath2010shortcuts}) and the resulting short-cuts can be extended to the partitioning principle and the determination of computationally efficient algorithms for the calculation of the resulting SCI. To this end, we will introduce weaker and stronger concepts of consonance for the partitioning principle that are sufficient to derive feasible and efficient algorithms for the calculation of the resulting SCI. 

The remainder of this work is organized as follows. In Section~\ref{sec:MTasSCI} we briefly review multiple tests with FWER control as well as SCI, and illustrate their close formal relationship. In Section~\ref{sec:CTandPP} we review the closed testing and partitioning principles for testing multiple hypotheses and show their computational equivalence. In Section~\ref{sec:ccfpp} we consider the case of (only) discrete parameters and introduce a simple algorithm for the partioning principle that applies under a weak version of consonance and the assumption that the parameters are bounded from below (e.g.\ are non-negative) or have a lower bound that can be determined from the data initially. Section~\ref{sec:CCandAfCP} considers the case of continuous parameter and extends the in Section~\ref{sec:ccfpp} introduced weak version of consonance for discrete parameter to a version (called weak uniform-consonance) for continuous parameter that takes into account the uncountable nature of the parameter space. 
We also introduce in this section an algorithm in which weak uniform-consonance is verified in a stepwise manner, leading to a conservative approximation of the lower confidence bounds. The algorithm is shown to apply to a generalization of the weighted Holm procedure to continuos parameter. Since the algorithm is not always feasible or may come with significant computational efforts, we suggest in 
Section~\ref{sec:monWeightBonf} 
an alternative, more efficient algorithm for the determination of a conservative approximation of the SCI bounds for continuous parameter that works under more general assumptions as the ones made in previous work on informative SCI (see e.g. \citealp{schmidt2014informative, schmidt2015informative}, \citealp{informativeSCIBrannathKlugeScharpenberg} and \citealp{klugeBrannath2026GSD}). The paper concludes with a short summary and discussion where also future research topics are indicted. The more complex mathematical results and proofs are presented in the Appendix.

\section{Multiple tests as one-sided simultaneous confidence intervals}\label{sec:MTasSCI}
We assume $m$ variationally independent parameter $\theta_i\in\Theta_i$ for a non-empty $\Theta_i\subseteq \mathbb{R}$,  $i\in  I:=\{1,\ldots,m\}$, 
with the $m$ null hypotheses 
\begin{align}\label{eq:hyps}
    H_1\subset \Theta_1, \ldots, H_m \subset \Theta_m .
\end{align} 
Variational independence of the parameters means that every vector $\theta=(\theta_1,\ldots,\theta_m)\in \Theta:=\times_{i=1}^m \Theta_i\subseteq \mathbb{R}^m$ represents a possible and unique parameter constellation. This implies that $\cap_{i\in J} H_j\not=\emptyset$ for all $J\subseteq I$, a situation which is often denoted by `unrestricted' or `free' hypotheses. Note that (with a slight abuse of notation) we also understand each $\Theta_i$ and $H_i$ as subsets of $\Theta$. A multiple test for \eqref{eq:hyps} can be represented by the vector valued test decision function 
\begin{align*}\varphi=(\varphi_1,\ldots,\varphi_m)\in \{0,1\}^m
\quad\text{where}\quad \varphi_i= \left\{\begin{array}{ll} 0 & \text{ retain }  H_i  \\
1 & \text{ claim } K_i:=\Theta_i\setminus H_i.
\end{array} \right.
\end{align*}
Often, one aims on `strong control' of the `familywise error rate' (FWER) at level $\alpha$, which means that:  
$$\mathbb{P}_\theta\Big(\varphi_i=1 \text{ for at least one }i\in\{1,\ldots,m\}\text{ with }\theta_i\in H_i\Big) \le \alpha \quad\text{for all } \theta\in \cup_{i=1}^m H_i.$$

A generally more informative approach for multiple inference is to calculate a simultaneous confidence interval for $\theta\in\Theta$, e.g.\ an one-sided simultaneous interval
$\text{SCI}=\times_{i=1}^m [L_i,\infty)$ with lower bounds $L_i$ that satisfy
$\mathbb{P}_\theta\Big(L_i\le \theta_i \text{ for all }i\in\{1,\ldots,m\}\Big) \ge 1-\alpha\text{ for all }\theta\in\Theta,$
or equivalently
\begin{align}\label{eq:sci}
\mathbb{P}_\theta\Big(L_i> \theta_i \text{ for at least one }i\in\{1,\ldots,m\}\Big) \le \alpha \quad\text{for all }\theta\in\Theta.
\end{align}

It is interesting to note that a multiple test with FWER control can formally be considered as a lower one-sided confidence interval with coverage probability $\ge 1-\alpha$ for the (less informative) vector of parameters  
\begin{align*} \delta=(\delta_1,\ldots,\delta_m)\in \{0,1\}^m  \quad\text{where}\quad \delta_i= \left\{\begin{array}{ll} 0 & \text{ if } \theta_i\in H_i  \\
                               1 & \text{ if } \theta_i\in K_i,
\end{array} \right.\end{align*}
which indicate the true and false null hypotheses for the (unknown) true parameter $\theta$. The familywise error rate and its strong control can then be written in terms of this parameter as 
\begin{align}\label{eq:fwer_del}
\mathbb{P}_\delta\Big(\varphi_i> \delta_i \text{ for at least one }i\in\{1,\ldots,m\}\Big) \le \alpha \text{ for all }
\delta\in\{0,1\}^m.
\end{align}
Comparing this to \eqref{eq:sci} shows that the multiple test is equivalent to the SCI for $\delta$ with lower bounds $L_i=\varphi_i$. 

We should note here that the probability in \eqref{eq:fwer_del} indexed by $\delta$ stands for the
supremum of the same probabilities indexed by  $\theta$ (i.e.\ $\mathbb{P}_\theta$) for  $\theta\in \Theta_\delta:=(\cap_{i:\delta_i=0} H_i)\cap (\cap_{i:\delta_i=1} K_i)$.

\section{Closed testing and partitioning principle}\label{sec:CTandPP}
The closed testing principle  for $H_1, \ldots, H_m$ can also be phrased in terms of the binary parameter vector~$\delta$: A closed test requires test decision functions
$\phi_\delta$ for all $\delta \in \{0,1\}^m$ with the property that 
\begin{align}\label{eq:ec_clt}
\max_{\delta'\le \delta}\mathbb{P}_{\delta'}(\phi_{\delta}=1)\le \alpha,
\end{align}
whereby we set $\phi_\delta=0$ for the parameter $\delta=\mathbf{1}=(1,\ldots,1)$ that indicates $\theta\in\cap_{i=1}^m K_i$, and define $\phi_{\delta}=1$ if $\Theta_\delta=\emptyset$.
As in \eqref{eq:fwer_del}, we need to understand in \eqref{eq:ec_clt} the probabilities indexed by $\delta'$ as the supremum of the same probabilities over $\theta\in\Theta_{\delta'}$.
Because $\delta'\le \delta$ is equivalent to $\theta_i\in H_i$ for $\delta_i=0$ and $\theta_j\in \Theta_j=H_j\cup K_j$ (unspecified) for $\delta_j=1$, equation \eqref{eq:ec_clt} simply means that $\phi_\delta$ is a (local) level $\alpha$ test for the intersection hypothesis $H_\delta:=\cap_{i:\delta_i=0} H_i$. 

The individual null hypothesis $H_i$ can be represented by the set $\{\delta':\delta'_i=0\}\subseteq\{0,1\}^m$ 
and the multiplicity adjusted test decision function for $H_i$, which results from the closed testing principle with local level $\alpha$ tests $(\phi_\delta)_{\delta\in \{0,1\}^m}$, is given by 
\begin{align} \label{eq:ctp}
\varphi^{\text{closed}}_i = \min_{\delta':\delta_i'=0} \phi_{\delta'}.
\end{align}
The test decision $\varphi^{\text{closed}}_i=1$ means to reject $H_i$ if and only if all intersection hypothesis $H_{\delta'}\subseteq H_i$  are rejected with $\phi_{\delta'}=1$. The closed test can be extended to also test the intersection hypotheses  $H_\delta$, 
$\delta\in\{0,1\}^m$, namely by the test decision functions  $\varphi^{\text{closed}}_\delta=\min_{\delta'\le\delta} \phi_{\delta'}$ where $\delta'\le \delta$ is to be understood component wise, i.e.\ $\delta'_i\le \delta_i$ for all $i=1,\ldots,m$.

Let us turn now to the partition principle. 
According to \citet{finner2002partitioning}, the `(weak) partitioning principle' considers the natural partition of $\Theta$ which we have encoded by $\delta$ in Section~\ref{sec:MTasSCI}, namely $(\Theta_\delta)_{\delta\in\{0,1\}^m}$ with
$\Theta_\delta= (\cap_{i:\delta_i=0} H_i)\cap (\cap_{i:\delta_i=1} K_i)$. For the sake of notational simplicity, we have added the joint alternative $\Theta_{\mathbf{1}}=\cap_{i=1,\ldots,m} K_i$ to the natural partition of \citet{finner2002partitioning}. 
For the weak partitioning  principle, we have to define for each $\Theta_\delta$ a decision function $\psi_\delta$ with the property 
\begin{align}\label{eq:ec_wpp}
\mathbb{P}_{\delta}(\psi_{\delta}=1)\le \alpha,
\end{align}
whereby we set $\psi_{\mathbf{1}}=0$ and $\psi_\delta=1$ if $\Theta_\delta=\emptyset$. Again, we have to understand the rejection probabilities indexed by $\delta$ as suprema over $\theta \in \Theta_\delta$. The (weak) partitioning principle then uses 
for $H_i=\dot{\bigcup}_{\delta':\delta_i'=0}\Theta_{\delta'}$ the test decision function 
\begin{align}\label{eq:ptp}
    \varphi_i^{\text{part}}=\min_{\delta':\delta_i'=0}\psi_{\delta'},
\end{align}
which simply checks, whether all $\Theta_{\delta'}\not=\emptyset$ are rejected whose union yields the hypothesis $H_i$. Like for the closed test, the intersection $H_\delta$ can be tested with the decision function
$\varphi^{\text{part}}_\delta=\min_{\delta'\le\delta} \psi_{\delta'}$.

The formal identity of \eqref{eq:ctp} and \eqref{eq:ptp} asks for a common understanding of the partitioning and closed testing principles. This is achieved by understanding a multiple test (with strong FWER control) as SCI for 
$\delta$, as described in the previous section. Indeed, by \eqref{eq:ec_clt} and \eqref{eq:ec_wpp} we can understand the sets 
\begin{align*}
C_{\text{closed}}=\{\delta'\in\{0,1\}^m:\phi_{\delta'}=0\}\quad \text{and}\quad C_{\text{part}}=\{\delta'\in\{0,1\}^m:\psi_{\delta'}=0\}
\end{align*}
as confidence sets for the parameter $\delta\in\{0,1\}^m$. Given this, the application of the test decision functions in \eqref{eq:ctp} and \eqref{eq:ptp} can be interpreted as projections of $C_{\text{closed}}$ and $C_{\text{part}}$ on the
lower simultaneous intervals 
\begin{align*}
SCI_{\text{closed}}:=\times_{i=1}^m [\varphi_i^{\text{closed}},1]\quad \text{and}\quad 
SCI_{\text{part}}=\times_{i=1}^m [\varphi_i^{\text{part}},1],
\end{align*}
whereby the `projection' is formally defined as the smallest one-sided simultaneous interval that contains
the confidence sets $SCI_{\text{closed}}$ and $SCI_{\text{part}}$, respectively.

The relationship between the closed and partitioning principle is illustrated in Figure~\ref{fig:ClosedVSPP} for the case of two hypotheses: A rejection of $H_2$ with the closed test requires a rejection of the parameters $\delta=(0,0)$ and $\delta=(1,0)$ (red squares) via $\varphi_2^{\text{closed}}=\min(\phi_{(0,0)},\phi_{(1,0)})=1$. This is the case if and only if $C_{\text{closed}}\subseteq\{(0,1),(1,1)\}$. A rejection of $H_2$ with the partitioning principle requires equally the rejection of the parameters $\delta=(0,0)$ and $\delta=(1,0)$ (red squares) via $\varphi_2^{\text{part}}=\min(\psi_{(0,0)},\psi_{(1,0)})=1$ or equivalently $C_{\text{part}}\subseteq\{(0,1),(1,1)\}$. The local test $\phi_{(1,0)}$ of the closed test needs to control the type I error under $\delta\in\{(0,0),(1,0)\}$ (blue ellipse in left picture). The local test $\psi_{(1,0)}$ for the partitioning principle must control the type  I error only under $\delta=(1,0)$ (blue circle in right picture). The latter can lead to more efficient tests, while the `algorithmic' requirement of rejecting $(0,0)$ and $(1,0)$ for making an individual claim on $\theta_2$ is the same for both principles.

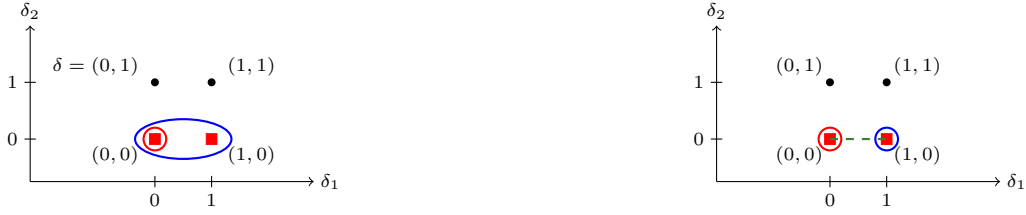
\begin{figure}[h]\scriptsize
\begin{minipage}{0.45\textwidth}
\begin{center}
    \begin{tikzpicture}[scale=0.75]
        \draw (0,0) node[draw=red, rectangle, fill=red, minimum size=4pt, inner sep=0pt]{} node[below left] {$(0,0)\hspace{.5em}$}; %{$\Theta_{00}=H_1\cap H_2$};
        \fill (0,1) circle (2pt) node[above left]  {$\delta=(0,1)\hspace{.5em}$}; %{$\Theta_{01}=H_1\cap K_2$}; 
        \draw (1,0) node[draw=red, rectangle, fill=red, minimum size=4pt, inner sep=0pt]{} node[below right]  {$\hspace{.5em}(1,0)$}; %{$\Theta_{10}=K_1\cap H_2$};
        \fill (1,1) circle (2pt) node[above right] {$\hspace{.5em}(1,1)$}; %{$\Theta_{11}=K_1\cap K_2$};
        \draw[red, thick] (0,0) ellipse (0.2 and 0.2);
        \draw[blue, thick] (0.5,0.0) ellipse (0.85 and 0.35);
        \draw[->] (-2.2,-0.75) -- (2.8,-0.75) node[right] {$\delta_1$};
        \draw (0,-0.75+0.1) -- (0,-0.75-0.1) node[below] {$0$};
        \draw (1,-0.75+0.1) -- (1,-0.75-0.1) node[below] {$1$};
        \draw[->] (-2.2,-0.75) -- (-2.2,2) node[above] {$\delta_2$};
        \draw (-2.2+0.1,0) -- (-2.2-0.1,0) node[left] {$0$};
        \draw (-2.2+0.1,1) -- (-2.2-0.1,1) node[left] {$1$};
    \end{tikzpicture}
    \end{center}
\end{minipage}
\hfill
\begin{minipage}{0.45\textwidth}\scriptsize
\begin{center}
    \begin{tikzpicture}[scale=0.75]
        \draw (0,0) node[draw=red, rectangle, fill=red, minimum size=4pt, inner sep=0pt]{} node[below left] {$(0,0)$}; %{$\Theta_{00}=H_1\cap H_2$};
        \draw[red, thick] (0,0) ellipse (0.2 and 0.2);
        %\fill (0,0) circle (2pt) node[below left] {$H_1\cap H_2=(0,0)$};
        \fill (0,1) circle (2pt) node[above left] {$(0,1)$}; %{$\Theta_{01}=H_1\cap K_2$};
        \draw (1,0) node[draw=red, rectangle, fill=red, minimum size=4pt, inner sep=0pt]{} node[below right] {$(1,0)$}; 
        \draw[blue, thick] (1,0) ellipse (0.2 and 0.2);%{$\Theta_{10}=K_1\cap H_2$};
        %\fill (1,0) circle (2pt) node[below right] {$(1,0)$};
        \fill (1,1) circle (2pt) node[above right] {$(1,1)$}; %{$\Theta_{11}=K_1\cap K_2$};
        %\draw[dashed, OliveGreen, thick] (0,0) -- (1,0);
        \draw[->] (-2,-0.75) -- (3,-0.75) node[right] {$\delta_1$};
        \draw (0,-0.75+0.1) -- (0,-0.75-0.1) node[below] {$0$};
        \draw (1,-0.75+0.1) -- (1,-0.75-0.1) node[below] {$1$};
        \draw[->] (-2,-0.75) -- (-2,2) node[above] {$\delta_2$};
        \draw (-2+0.1,0) -- (-2-0.1,0) node[left] {$0$};
        \draw (-2+0.1,1) -- (-2-0.1,1) node[left] {$1$};
       \draw[dashed, OliveGreen, thick] (0,0) -- (1,0);
    \end{tikzpicture}
    \end{center}
\end{minipage}
\caption{Visualization of the closed testing principle (left) and the partitioning principle (right) for the case of two hypotheses $H_1$ and $H_2$. \label{fig:ClosedVSPP}}
\end{figure}

Based on this common interpretation of the closed and partitioning principle, we may ask whether `consonance' concepts for closed tests that imply computational shortcuts (see e.g.
\citealp{hommel2007powerful}; \citealp{brannath2010shortcuts}) can be extended to `consonance' concepts for the partitioning principle that imply numerically efficient algorithms for the computation of simultaneous confidence intervals. This question is particularly important for continuous parameters and already interesting for (non-binary) discrete parameters that (e.g.) better encode nested hypotheses. This paper is devoted to this question and will provide positive answers. 

We end the section with a remark on conditions \eqref{eq:ec_clt} and \eqref{eq:ec_wpp}. Since the first means to control the type I error rate over a larger parameter subset than the second, the second has the potential to provide more efficient tests than the first. This was illustrated in Figure~\ref{fig:ClosedVSPP} for the case of two hypotheses by the ellipses indicating the parameter constellations $\delta\in\{0,1\}^2$ under which the type I error must be controlled. \citet{finner2002partitioning}  mathematically prove and illustrate by examples, that using the partitioning principle never leads to less efficient tests than the closed testing principle and can sometimes provide more efficient tests. Let for all $i$ the indicator vector $\delta^{(i)}$ defined by $\delta^{(i)}_i=0$, and $\delta_j^{(i)}=1$ for $j\neq i$. Then, if the property $\varphi_i^{\text{part}}=\psi_{\delta^{(i)}}$ is satisfied for all $i$, no efficiency loss results from using the minimum in \eqref{eq:ptp} instead of the local test $\psi_{\delta^{(i)}}$ for $\Theta_{\delta^{(i)}}$, which can result in a further efficiency gain.
Accordingly, \citet{finner2002partitioning} suggest to speak of the `strong' partitioning principle 
when $\varphi_i^{\text{part}}=\psi_{\delta^{(i)}}$ for all $i$. 

Because the focus here is on aspects of computational rather than statistical efficiency, we can largely ignore the difference between \eqref{eq:ec_clt} and \eqref{eq:ec_wpp}, and need also not to distinguish between the \textsl{weak} and \textsl{strong} partitioning principles. Since \eqref{eq:ec_wpp} is more general than \eqref{eq:ec_clt}, we assume from now on local tests $\psi_{\delta}$, $\delta\in\{0,1\}^m$ that satisfy \eqref{eq:ec_wpp}, and denote the underlying method as `partitioning principle'. However, one should note that by the formal identity of \eqref{eq:ctp} and \eqref{eq:ptp}, the below developed algorithms apply to the closed testing and both partitioning principles.

\section{Weak consonance and related algorithm for the partitioning principle}\label{sec:ccfpp}

Consonance is a helpful property of closed testing procedures that -- if  satisfied -- enables an efficient implementation of the procedure by substantially reducing the number of required intersection tests
(\citealp{hommel2007powerful}; \citealp{brannath2010shortcuts}). It has originally been defined by \citet{gabriel1969simultaneous}; see also \citet{finner2008}. Loosely speaking, consonance of the closed test means that the rejection of any intersection hypothesis implies the rejection of at least one individual hypothesis from the intersection. With our indicator variable $\delta$, this can be formalized as follows: 
For all $\delta\in \{0,1\}^m$ the rejection of $H_\delta$, i.e.\ 
$\varphi^{\text{closed}}_\delta=\min_{\delta'\le\delta} \phi_{\delta'}=1$, implies the existence of $i\in I$ with $\delta_i=0$ such that $H_i$ can be rejected, i.e. 
\begin{align}\label{eq:cons}
\varphi^{\text{closed}}_i=\min_{\delta':\delta'_i=0} \phi_{\delta'}=1. % Zweidimesionales Bild
\end{align}

We introduce here a weaker version of consonance that, as we will see later, already provides efficient projection algorithms also for non-binary discrete parameters. Of course, the weaker the version, the easier it is to verify. 
\begin{Def}\label{eq:weak_cons}
We call a test procedure based on the partitioning (or closed testing) principle with local tests $(\psi_\delta)_{\delta\in \{0,1\}^m}$ `weakly consonant' if  for all $\delta\in\{0,1\}^m$ the event $\varphi_\delta=\min_{\delta'\le\delta} \psi_{\delta'}=1$ implies the existence of $i\in I$ with $\delta_i=0$ such that 
\begin{align}\label{eq:G_delta}
\min_{\delta'\in G^{(i)}_\delta} \psi_{\delta'}=1
\quad\text{ for }\quad
G^{(i)}_\delta:=\{\delta':\delta'_i=0\text{ and }\delta'_j\ge \delta_j\text{ for all }j\not=i\}\subseteq \{0,1\}^m. 
\end{align}
Obviously,  
$G^{(i)}_\delta$ is smaller than $\{\delta':\delta'_i=0\}$ in the minimum of \eqref{eq:cons} whenever $\delta\not=\mathbf{0}=(0,\ldots,0)$.
\end{Def}

With only two parameters ($m=2$), weak consonance and consonance are equivalent, because e.g.\ the rejection of $\delta_1'\leq\delta_1$, $\delta_2'\leq\delta_2$ and $\delta_2'\geq\delta_2$ implies that $H_0:\delta_1'\leq\delta_1$ can be rejected. 
%This is because if, for e.g. $\delta=(0,1)$ it holds $\varphi_{\delta}=\min(\psi_{(0,0)},\psi_{(0,1)})=1$ the rejection of $\delta=(0,1)$ (which is the only parameter contained in $\{\delta':\delta_1'=0\}$ but not in $G_{(0,1)}^{(i)}$) is already covered.
However, this equivalence is no longer true for $m\ge 3$. 
In Figure~\ref{fig:ComparisonConsonanceVSWeaklyConsonance},  consonance (left plot) is compared to weak consonance (right plot) for three hypotheses $H_1$, $H_2$, and $H_3$. In both plots, the intersection $H_1\cap H_2$ can be rejected (by the closed testing or partitioning principle), since $(0,0,0)$ and $(0,0,1)$ are rejected locally
(big red squares). Consonance requires that there exists a component $i\in\{1,2\}$, such that all parameters $\delta'$ with $\delta'_i=0$ can be locally rejected. In the left plot this applies to $i=2$, since also $(1,0,0)$ and $(1,0,1)$ are locally rejected (small red squares). In the 
the right plot, with weak consonance, the local rejection of only $(1,0,1)$ is required, because  $G_{(0,0,1)}^{(2)}=\{(0,0,1),(1,0,1)\}$. 

% Dreidiemsionales Bild
\begin{figure}[h]\scriptsize
\begin{minipage}{0.45\textwidth}
\begin{center}
\begin{tikzpicture}[scale=1,
    x={(1cm,0cm)},
    y={(0.6cm,0.3cm)},
    z={(0cm,1cm)}]

% Würfelkanten (hinten gestrichelt)
\draw[dashed] (0,0,0) -- (0,1,0);
\draw[dashed] (0,1,0) -- (0,1,1);
\draw[dashed] (0,1,0) -- (1,1,0);

% Würfelkanten (vorne)
\draw (0,0,0) -- (1,0,0) -- (1,0,1) -- (0,0,1) -- cycle;

% Würfelkanten (oben)
\draw (0,0,1) -- (1,0,1) -- (1,1,1) -- (0,1,1) -- cycle;

% Würfelkanten (rechts)
\draw (1,0,0) -- (1,1,0);
\draw (1,1,0) -- (1,1,1);

% Alle Eckpunkte (schwarz)
\foreach \x in {0,1}
\foreach \y in {0,1}
\foreach \z in {0,1}
{
    \fill (\x,\y,\z) circle (2pt);
}
% Ablehnungen
\draw (0,0,0) node[draw=red, rectangle, fill=red, minimum size=7pt, inner sep=0pt]{} node[below left] {$(0,0,0)\hspace{.5em}$};
\draw (1,0,0) node[draw=red, rectangle, fill=red, minimum size=4pt, inner sep=0pt]{} node[below right] {$(1,0,0)\hspace{.5em}$};
\draw (1,0,1) node[draw=red, rectangle, fill=red, minimum size=4pt, inner sep=0pt]{} node[right] {$(1,0,1)$};
\draw (0,0,1) node[draw=red, rectangle, fill=red, minimum size=7pt, inner sep=0pt]{} node[above left] {$\delta=(0,0,1)\hspace{.5em}$};
\end{tikzpicture}

    \end{center}
\end{minipage}
\hfill
\begin{minipage}{0.45\textwidth}\scriptsize
\begin{center}
\begin{tikzpicture}[scale=1,
    x={(1cm,0cm)},
    y={(0.6cm,0.3cm)},
    z={(0cm,1cm)}]

% Würfelkanten (hinten gestrichelt)
\draw[dashed] (0,0,0) -- (0,1,0);
\draw[dashed] (0,1,0) -- (0,1,1);
\draw[dashed] (0,1,0) -- (1,1,0);

% Würfelkanten (vorne)
\draw (0,0,0) -- (1,0,0) -- (1,0,1) -- (0,0,1) -- cycle;

% Würfelkanten (oben)
\draw (0,0,1) -- (1,0,1) -- (1,1,1) -- (0,1,1) -- cycle;

% Würfelkanten (rechts)
\draw (1,0,0) -- (1,1,0);
\draw (1,1,0) -- (1,1,1);

% Alle Eckpunkte (schwarz)
\foreach \x in {0,1}
\foreach \y in {0,1}
\foreach \z in {0,1}
{
    \fill (\x,\y,\z) circle (2pt);
}
% Ablehnungen
\draw (0,0,0) node[draw=red, rectangle, fill=red, minimum size=7pt, inner sep=0pt]{} node[below left] {$(0,0,0)\hspace{.5em}$};
\fill (1,0,0) circle (2pt) node[below right] {$ (1,0,0) $};;
\draw (1,0,1) node[draw=red, rectangle, fill=red, minimum size=4pt, inner sep=0pt]{} node[right] {$(1,0,1)$};
\draw (0,0,1) node[draw=red, rectangle, fill=red, minimum size=7pt, inner sep=0pt]{} node[above left] {$\delta=(0,0,1)\hspace{.5em}$};
\end{tikzpicture}
    \end{center}
\end{minipage}
\caption{Consonance of closed test (left) and weak consonance for partitioning principle (right).\label{fig:ComparisonConsonanceVSWeaklyConsonance}}
\end{figure}
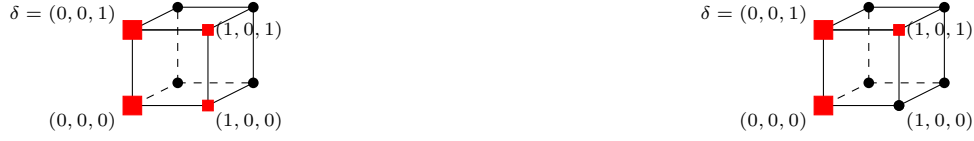

\citet{hommel2007powerful} consider a general class of weighted Bonferroni tests with a kind of monotonicity property for the weights, for which one can easily show that they are weakly consonant. We review this class using our binary parameter $\delta$. To this end, assume an undadjusted p-value $p_i$ for each $H_i$ and weights $w_i(\delta)$ that depend (in general) on the whole vector $\delta$, whereby $w_i(\delta)=0$ whenever $\delta_i=1$, since we are not interested in rejecting the alternative $K_i$. We test each intersection hypothesis with the adjusted p-value $p(\delta)=\min_{1=1}^m p_i/w_i(\delta)$ (that could, but need not to be truncated at 1) and reject $H_\delta$ if $p(\delta)\le \alpha$. \citet{hommel2007powerful} assume that each weight $w_i(\delta)$ is non-decreasing in $\delta_j$ for each $j\not=i$, i.e.\ $w_i(\delta')\ge w_i(\delta)$ whenever $\delta'_i=\delta_i$
and $\delta'_j\ge \delta_j$ for all $j\not=i$. Obviously, $p(\delta)=\min_{1=1}^m p_i/w_i(\delta)\le\alpha$ implies $p_i/ w_i(\delta)\le \alpha$ for at least on $i$ with $\delta_i=0$. By the assumed monotonicity property of $w_i(\delta)$, we obtain
$p(\delta')\le p_i / w_i(\delta')\le  p_i / w_i(\delta) \le\alpha$ for $\delta'$ with $\delta'_i=0$ and $\delta'_j\ge \delta_j$ for all $j\not=i$. This clearly implies weak consonance. 

As shown in \citet{hommel2007powerful} and \citet{brannath2010shortcuts}, consonant closed tests can be efficiently implemented by so-called `step-down' algorithms with at most $m$ intersection tests. %We show in the appendix, 
One can show that the same algorithm is possible whenever the local tests $\phi_\delta$ or $\psi_\delta$ satisfy the above introduced weak consonance property. In the next subsection, we present a generalization of 
this algorithm to (also) non-binary discrete parameters 
that applies under a generalization of the weak consonance property. 

\subsection{Weak consonance and efficient projection algorithm for discrete parameter}\label{sec:WCandAfDP}

We consider now the situation where $\Theta_i=\{\theta_{i0},\theta_{i1},\theta_{i2},\ldots\}$ for ordered discrete values $\theta_{i0}<\theta_{i1}<\theta_{i2}<\dots$ and we are interested in testing $H_{il}:\theta_i=\theta_{il}$ %versus  $K_{il}:\theta_i>\theta_{il}$ 
for all $\theta_{il}\in\Theta_i$ and $i\in I =\{1,\ldots,m\}$. Following the partitioning principle, we define for each $\theta\in \Theta=\times_{i=1}^m\Theta_i$ a level $\alpha$ test $\psi_\theta$ for $\theta$ which satisfies $\mathbb{P}_\theta(\psi_\theta )\le \alpha$. This provides the $1-\alpha$ confidence set
\begin{align}\label{eq:C}
C=\{\theta\in\Theta: \psi_\theta=0\}.
\end{align}
In general, this is not a simultaneous interval and needs to be projected to the smallest one-sided simultaneous interval $SCI:=\times_{i=1}^m [L_i,\infty)\supseteq C$, if we are interested in an individualized inference for the components $\theta_i$. Since $\mathbb{P}_\theta\big(\theta\in SCI\big)\ge \mathbb{P}_\theta(\theta\in C)\ge 1-\alpha$, this is a simultaneous $1-\alpha$-confidence interval. 

The lower bounds $L_i$ of the projection can be determined with the in $l$ non-increasing decision functions
\begin{align}\label{eq:proj_disc}
    \varphi^{\text{proj}}_{il} = \min_{\theta': \theta_i'\le \theta_{il}} \psi_{\theta'}
\quad\text{ as }\quad
L_i := \min_{l=0,1,2,\ldots}\{\theta_{il}: \varphi^{\text{proj}}_{il} =0\}.
\end{align}
Note the formal similarity between  the above $\varphi^{\text{proj}}_{il}$ and $\varphi_i^{\text{part}}$ in \eqref{eq:ptp}. 
We are interested in algorithms that efficiently implement the two minima in \eqref{eq:proj_disc} and call any such algorithm a `projection algorithm'. 

In order to identify situations where an efficient projection algorithm exists (and to define it), we generalize the weak consonance property of the previous section. 
\begin{Def}\label{eq:weak_cons_disc}
We call a test procedure based on the partitioning principle with local tests $(\psi_\theta)_{\theta\in \Theta}$ `weakly consonant' if for all $\theta\in\Theta$ the event $\min_{\theta'\le\theta} \psi_{\theta'}=1$ implies the existence of $i\in I$ with  
\begin{align}\label{eq:G_Geil}
\min_{\theta'\in G^{(i)}_\theta} \psi_{\theta'}=1
\quad\text{ for }\quad
G^{(i)}_\theta:=\{\theta'\in \Theta:\theta'_i=\theta_i \text{ and } \theta'_j\geq\theta_j \text{ for all }j\not=i\}\subseteq \Theta.
\end{align}
Note the similarity between $G^{(i)}_\theta$ and $G^{(i)}_\delta$ in \eqref{eq:G_delta}. 
\end{Def}

With weak consonance, we get the following algorithm to calculate the projection of $C$ in \eqref{eq:proj_disc}:

\begin{algorithm}[H]
\setstretch{1.1}
\caption{Projection algorithm for discrete parameter} \label{alg_abstract_discrete}
Initialize the lower bounds to $\lambda=(\theta_{10},\dots,\theta_{m0})$
and set $k_i\rightarrow 0$ for all $1\leq i\leq m$\;
\While{$\psi_{\lambda}=1$}{
find $i$ %direction with known weak consonance
with $\min_{\theta'\in G^{(i)}_{\lambda}} \psi_{\theta'}=1$\; and 
%}{
step to the next parameter point in the $i$-th coordinate direction, i.e.:\\
update $\lambda_i\rightarrow\theta_{i,k_i+1}$ and $k_i\rightarrow k_i+1$\;
}
\Return{the lower confidence bounds $L_1=\lambda_1, \ldots, L_m=\lambda_m$.}
\end{algorithm}
\ \\[.25em]
The algorithm is illustrated in Figure \ref{fig:ProjectionAlgoDiscrete} for $m=2$.
It generalizes the well-known step-down algorithm for consonant closed tests (see \citealp{hommel2007powerful}).
We show in the Appendix that in each step with $\psi_\lambda=1$ we obtain
$\lambda_j\leq  L_j$ for all $j\in I$. Therefore $\psi_{\lambda}=1$ implies $\min_{\theta'\le\lambda} \psi_{\theta'}=1$, and by the weak consonance property, we find some direction $i$ with $\min_{\theta'\in G^{(i)}_{\lambda}} \psi_{\theta'}=1$. Hence, the algorithm steps forward whenever $\psi_\lambda=1$, and it stops when $\psi_\lambda=0$ for the first time. The latter implies $\lambda\in C$, and therefore the algorithm attends the projection $L$. 

Algorithm 1 is linear in $m$, in the sense that the number of steps (and performed local tests) is bounded by $\sum_{i=1}^m |L_i-\theta_{i,0}|$. The total number of local tests that would have to be performed without such an algorithm is either infinite (whenever a parameter has infinitely many values) or increases exponentially with $m$. 

\begin{figure}[h]
    \centering
% Obere Reihe
    \begin{minipage}{0.45\textwidth}
        \centering
        \begin{tikzpicture}[scale=0.6]

  \draw[->] (0,0) -- (5.5,0) node[right] {$\theta_1$};
  \draw[->] (0,0) -- (0,5) node[above] {$\theta_2$};
\def\circles{
  0/0, 0/1, 0/2, 0/3, 0/4,
  1/0, 1/1, 1/2, 1/3, 1/4,
  2/0, 2/3,  2/4,
  3/0, 3/1,
  4/0, 4/1, 
  5/0, 5/1, 5/2, 5/3
}
\foreach \x/\y in \circles {
    \fill (\x,\y) circle (2pt) {};
}

\def\rectangles{
  0/1, 0/2, 0/3, 0/4
}
\foreach \x/\y in \rectangles {
    \draw (\x,\y) node[draw=red, rectangle, fill=red, minimum size=4pt, inner sep=0pt]{};
}
\draw (0,0) node[draw=red, rectangle, fill=red, minimum size=7pt, inner sep=0pt]{};

\draw[dashed, OliveGreen, thick] (0,0) -- (0,4.5);
\draw[->, red, thick] (-0.25,0.25) -- (-0.25,0.75) {};

\node[text=blue, scale=1] at (2,2) {$\star$};
\node[text=blue, scale=1] at (2,1) {$\star$};
\node[text=blue, scale=1] at (3,2) {$\star$};
\node[text=blue, scale=1] at (4,2) {$\star$};
\node[text=blue, scale=1] at (3,4) {$\star$};
\node[text=blue, scale=1] at (3,3) {$\star$};
\node[text=blue, scale=1] at (4,3) {$\star$};
\node[text=blue, scale=1] at (5,4) {$\star$};
\node[text=blue, scale=1] at (4,4) {$\star$};
\end{tikzpicture}\\ 
        \text{Step 1}
    \end{minipage}
    \hfill
    \begin{minipage}{0.45\textwidth}
        \centering
         \begin{tikzpicture}[scale=0.6]
  \draw[->] (0,0) -- (5.5,0) node[right] {$\theta_1$};
  \draw[->] (0,0) -- (0,5) node[above] {$\theta_2$};
\def\circles{
  0/0, 0/1, 0/2, 0/3, 0/4,
  1/1, 1/2, 1/3, 1/4,
  2/3, 2/4,
  3/1,
  4/1,
  5/1, 5/2, 5/3
}
\foreach \x/\y in \circles {
    \fill (\x,\y) circle (2pt) {};
}

\def\rectangles{
  0/0, 0/1, 0/2, 0/3, 0/4,
  2/0, 3/0, 4/0, 5/0
}
\foreach \x/\y in \rectangles {
    \draw (\x,\y) node[draw=red, rectangle, fill=red, minimum size=4pt, inner sep=0pt]{};
}
\draw (1,0) node[draw=red, rectangle, fill=red, minimum size=7pt, inner sep=0pt]{};

\draw[dashed, OliveGreen, thick] (1,0) -- (5.5,0);
\draw[->, red, thick] (1.25,-0.25) -- (1.75,-0.25) {};

\node[text=blue, scale=1] at (2,2) {$\star$};
\node[text=blue, scale=1] at (2,1) {$\star$};
\node[text=blue, scale=1] at (3,2) {$\star$};
\node[text=blue, scale=1] at (4,2) {$\star$};
\node[text=blue, scale=1] at (3,4) {$\star$};
\node[text=blue, scale=1] at (3,3) {$\star$};
\node[text=blue, scale=1] at (4,3) {$\star$};
\node[text=blue, scale=1] at (5,4) {$\star$};
\node[text=blue, scale=1] at (4,4) {$\star$};
\end{tikzpicture}\\
        \text{Step 2}
    \end{minipage}

    \vspace{0.5cm}

    % Untere Reihe
    \begin{minipage}{0.45\textwidth}
        \centering
         \begin{tikzpicture}[scale=0.6]

  \draw[->] (0,0) -- (5.5,0) node[right] {$\theta_1$};
  \draw[->] (0,0) -- (0,5) node[above] {$\theta_2$};
\def\circles{
  0/0, 0/1, 0/2, 0/3, 0/4,
  1/0, 1/1, 1/2, 1/3, 1/4,
  2/0, 2/3, 2/4,
  3/0, 3/1,
  4/0, 4/1,  
  5/0, 5/1, 5/2, 5/3
}
\foreach \x/\y in \circles {
    \fill (\x,\y) circle (2pt) {};
}

\def\rectangles{
  0/0, 0/1, 0/2, 0/3, 0/4,
  1/0, 1/1, 1/2, 1/3, 1/4,
  2/0, 
  3/0, 
  4/0, 
  5/0
}
\foreach \x/\y in \rectangles {
    \draw (\x,\y) node[draw=red, rectangle, fill=red, minimum size=4pt, inner sep=0pt]{};
}
\draw (1,1) node[draw=red, rectangle, fill=red, minimum size=7pt, inner sep=0pt]{};

\draw[dashed, OliveGreen, thick] (1,1) -- (1,4);
\draw[->, red, thick] (0.75,1.25) -- (0.75,1.75) {};

\node[text=blue, scale=1] at (2,2) {$\star$};
\node[text=blue, scale=1] at (2,1) {$\star$};
\node[text=blue, scale=1] at (3,2) {$\star$};
\node[text=blue, scale=1] at (4,2) {$\star$};
\node[text=blue, scale=1] at (3,4) {$\star$};
\node[text=blue, scale=1] at (3,3) {$\star$};
\node[text=blue, scale=1] at (4,3) {$\star$};
\node[text=blue, scale=1] at (5,4) {$\star$};
\node[text=blue, scale=1] at (4,4) {$\star$};

\end{tikzpicture}\\
        \text{Step 3}
    \end{minipage}
    \hfill
    \begin{minipage}{0.45\textwidth}
        \centering
         \begin{tikzpicture}[scale=0.6]
  \draw[->] (0,0) -- (5.5,0) node[right] {$\theta_1$};
  \draw[->] (0,0) -- (0,5) node[above] {$\theta_2$};
\def\circles{
  0/0, 0/1, 0/2, 0/3, 0/4,
  1/0,
  2/3, 2/4,
  3/1,
  4/1, 
  5/1, 5/2, 5/3
}
\foreach \x/\y in \circles {
    \fill (\x,\y) circle (2pt) {};
}

\def\rectangles{
  0/0, 0/1, 0/2, 0/3, 0/4,
  1/0, 1/1, 1/2, 1/3, 1/4,
  2/0,  
  3/0, 
  4/0, 
  5/0
}
\foreach \x/\y in \rectangles {
    \draw (\x,\y) node[draw=red, rectangle, fill=red, minimum size=4pt, inner sep=0pt]{};
}
\node[text=blue, scale=2] at (2,1) {$\mathbf{\star}$};

\node[text=blue, scale=1] at (2,2) {$\star$};
\node[text=blue, scale=1] at (3,2) {$\star$};
\node[text=blue, scale=1] at (4,2) {$\star$};
\node[text=blue, scale=1] at (3,4) {$\star$};
\node[text=blue, scale=1] at (3,3) {$\star$};
\node[text=blue, scale=1] at (4,3) {$\star$};
\node[text=blue, scale=1] at (5,4) {$\star$};
\node[text=blue, scale=1] at (4,4) {$\star$};

\draw[dashed, OliveGreen, thick] (2,1) -- (2,4.5);
\draw[dashed, OliveGreen, thick] (2,1) -- (5.5,1);
\end{tikzpicture}\\
        \text{Step 4}
    \end{minipage}
    \caption{Projection algorithm for discrete parameter\label{fig:ProjectionAlgoDiscrete}}
\end{figure}
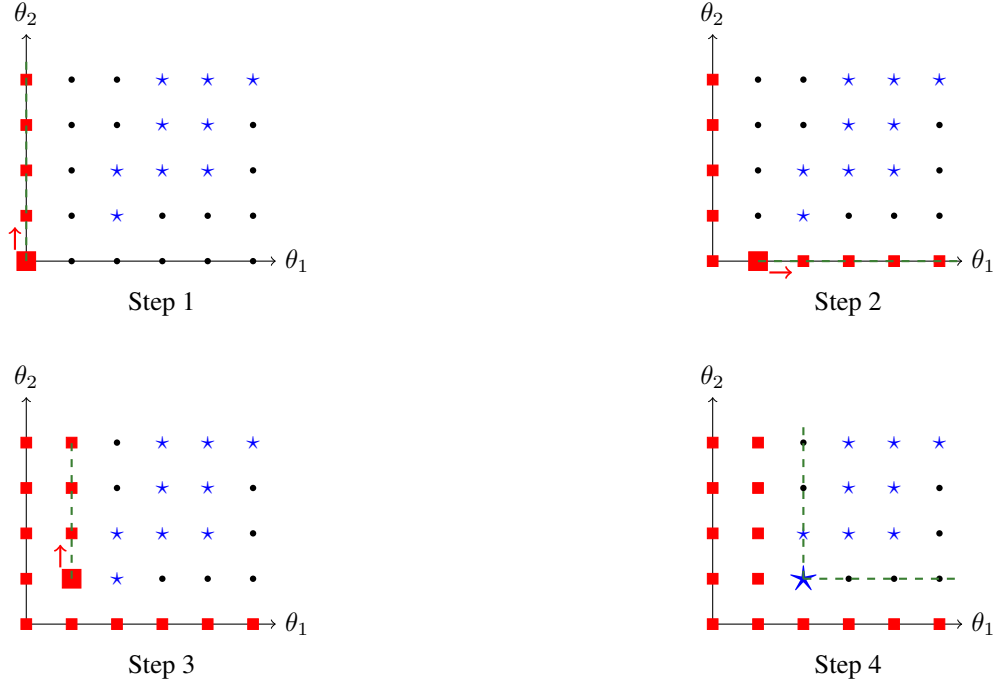

\subsection{Monotone weighted Bonferroni tests for discrete parameter}

As a general example, we define a class of local tests $\psi_\theta$, that generalizes the weighted Bonferroni tests of \citet{hommel2007powerful} from the binary to the non-binary case. To this end, we assume for each $\theta_{il}$ a p-value $p_i(\theta_{il})$ for testing 
$\theta_i=\theta_{il}$ (typically against $\theta_i>\theta_{il}$).
The local test for each $\theta\in\Theta$ is a weighted Bonferroni test with p-value $p(\theta)=\min_{i=1}^m p_i(\theta_i)/w_i(\theta)$ where the weights satisfy $\sum_{i=1}^m w_i(\theta)=1$, and each $w_i(\theta)$ is positive and non-decreasing in $\theta_j$ for all $j\neq i$. Weak consonance of the local tests $p(\theta)\le \alpha$ follows by the same arguments as for the original weighted Bonferroni tests of \citet{hommel2007powerful} reviewed in the previous subsection.   
\ \\[.5em]
\noindent \textbf{Example 1 -- Generalized weighted Holm procedure.} As a hypothetical application example, assume a clinical trial with non-inferiority and superiority hypotheses, $H_{i1}\subset H_{i2}$, for two endpoints $i=1,2$. We encode for each endpoint $i$ the two null hypotheses and the superiority alternative with the values 
$1,2$ and $3$. This gives the discrete parameter space $\Theta=\{1,2,3\}^2$. We assume that a non-inferiority claim is considered more important than the superiority claim, and therefore, we use for 
$\theta=(\theta_1,\theta_2)\in\{1,2,3\}^2$ weights $w_i(\theta)=a_i(\theta_i)/\big(a_1(\theta_1)+a_2(\theta_2)\big)$, $i=1,2$, based on non-increasing functions $a_i:\{1,2,3\}\to \mathbb{R}$ with $a_i(3)=0$ for all $i=1,2,3$, with the latter condition to avoid rejection of the superiority alternatives. A simple example of such a function would be $a_i(j)=3-j$. Obviously, each $w_i(\theta)$ is increasing in $\theta_j$ for $j\not=i$, and therefore the resulting procedure is weakly consonant. Hence, we can apply Algorithm 1.

The example and application of Algorithm 1 can easily be extended to situations with more than two endpoints and a finite or countably infinite sequence of nested hypothesis $H_{i1}\subset H_{i2}\subset H_{i3} \subset \cdots$ for each endpoint $i$ and non-increasing individual weighting functions $a_i:\mathbb{N}\to [0,\infty)$. 
This can be viewed as a generalization of the weighted Bonferroni-Holm procedure for multiple hypotheses, where the latter applies only to binary parameters with decreasing weight functions $a_i(0)=\alpha_i>0$ and $a_i(1)=0$. In Section~\ref{sec:cwholm} we will present a generalization of this procedure to continuous parameters.

\subsection{Initialization without origin}

We end the section with a comment on how to deal with (discrete) parameters that are not bounded from below. In this case, we need to know an initial parameter point $\theta_0=(\theta_{10},\ldots,\theta_{m0})$ that is smaller than the lower bound of the one-sided confidence interval, i.e.\ $\theta_{i0}\le L_i$ for all $i=1,\ldots,m$. The algorithm can then start from $\theta_0$ as described in Algorithm 1. The starting point $\theta_0$ will in general depend on the data and may require an extra initialization algorithm for its determination.

We are not aware of a general initialization algorithm for the determination of $\theta_0$
that follows from a sufficiently easy to verify consonance property. Such an algorithm needs to be developed on a case-by-case basis. For the above described class of weighted Bonferroni tests, an initial value can easily be determined, if each p-value $p_i(\theta_i)$ is increasing in $\theta_i$ (which is typically the case for one-sided p-values) and there exists a $\theta_i^{*}\in\mathbb{R}$ such that on the half-space $\{\theta\in\Theta:\theta_i\le \theta^{*}_i\}$, the weight $w_i(\theta)$ is  bounded from below by some positive value $w_{i0}>0$, i.e. $w_i(\theta)\geq w_{i0}>0$ for all $\theta\in\Theta$ with $\theta_i\leq\theta_i^*$. The initial value $\theta_{i0}$ can then be determined as the minimum of $\theta_i^*$ and the lower one-sided confidence bound at level $w_{i0}\,\alpha$.

\section{Consonance concepts and algorithms for continuous parameter}\label{sec:CCandAfCP}

We consider now the more complex situation with $m$ continuous parameters $\theta_i\in \Theta_i$ where each $\Theta_i\subseteq \mathbb{R}$ is an interval. 
As before we assume a family of local test $(\psi_\theta)_{\theta\in\Theta}$ with $\mathbb{P}_{\theta}(\psi_\theta)\le\alpha$ for all $\theta\in\Theta$. Hence, we consider the finest possible partitioning. We start again with the confidence set $C=\{\theta:\psi_\theta=0\}$ and aim to determine its projection, i.e. the smallest $SCI=\times_{i=1}^m [L_i,\infty)\supseteq C$. Similarly to the discrete case, the lower bounds $L_i$ of the projection are defined with
\begin{align}\label{eq:proj_cont}
    \varphi^{\text{proj}}_i(\theta_i) = \min_{\theta': \theta_i'\le \theta_i} \psi_{\theta'}
\quad\text{ as }\quad
L_i := \inf\{\theta'_i: \varphi^{\text{proj}}_{i}(\theta'_i) =0\}.
\end{align}
The major difference to \eqref{eq:proj_disc} for the discrete case is that, with a continuous parameter, for the bound $L_i$ we need to determine the infimum instead of the minimum, which is a mathematical limit and therefore can (in general) only be determined approximately. Moreover, there are no naturally defined values, the algorithm can run through. 

A more severe issue comes with the in general non-monotonous behavior of $\psi_\theta$ in $\theta$, by which we can never be sure that we are not jumping over $L_i$ when moving along any kind of grid, and so end up with an anti-conservative value. This issue 
can be avoided under assumptions that are stronger than the consonance and weak consonance conditions, which, as we will see later, still applies to typical examples. 

\begin{Def}
Assume a family of local tests $(\psi_\theta)_{\theta\in\Theta}$. We call this family `uniform-consonant at $\theta\in\Theta$', if
the event $\min_{\theta'\le\theta}\psi_{\theta'}=1$ implies the existence of $i\in I$ 
with  
\begin{align}\label{eq:_ucons}
\varphi^{\text{proj}}_i(\theta_i) = \min_{\theta': \theta'_i\le \theta_i}\psi_{\theta'}=1.
\end{align}
Note the difference to the requirement \eqref{eq:cons} for consonance, which only requires $\psi_{\theta'}=1$ for $\theta'$ with $\theta'_i=\theta_i$.

We call the family `weakly uniform-consonant at $\theta$', if the event $\min_{\theta'\le\theta}\psi_{\theta'}=1$ implies the existence of $i\in I$ with 
\begin{align}\label{eq:D_cons}
 \min_{\theta'\in D^{(i)}_\theta} \psi_{\theta'}=1
\quad\text{ for }\quad
D^{(i)}_\theta:=\{\theta'\in \Theta:\theta'_i\le\theta_i \text{ and } \theta'_j\geq\theta_j \text{ for all }j\not=i\}\subseteq \Theta.
\end{align}
Note that $D^{(i)}_\theta$ is larger than $G^{(i)}_\theta$, and so weak uniform-consonance is stronger than weak consonance.

We say that $(\psi_\theta)_{\theta\in\Theta}$ is uniform-consonant and  weakly uniform-consonant at $\theta$ `in the direction of $i$' if \eqref{eq:_ucons} respectively \eqref{eq:D_cons} is satisfied for $i$.
\end{Def}

By the formal definition \eqref{eq:proj_cont} of the projection $L_i$, uniform-consonance at $\theta$ \textsl{in the direction of} $i$ implies $L_i\ge \theta_i$, and the failure of this 
%uniform-consonance at $\theta$ \textsl{in the direction of} $i$ 
(i.e.\ $\varphi^{\text{proj}}_i(\theta_i)=0$) 
implies $L_i\le \theta_i$. Therefore, knowing whether a family of local tests is uniform-consonant at some point $\theta$ in direction $i$ provides either a lower or an upper bound for $L_i$. Below, we will provide an algorithm for which it is sufficient to know whether we have weak uniform-consonance or not in a given direction $i$ in order to either improve a lower or upper bound for $L_i$. 

In the next subsection, we will give an example of local tests that are weakly uniform-consonant at each $\theta$ and for which we can always decide whether we have weak uniform-consonance or not in any given direction.  \\

\subsection{Continuously weighted Holm procedure}\label{sec:cwholm}

Like in Example 1, we consider a clinical trial with two endpoints and corresponding efficacy parameter $\theta_i\in\mathbb{R}$, $i=1,2$, where for both endpoints larger values correspond to higher efficacy. We further assume for each endpoint $i$ and parameter value $\vartheta_i\in\Theta_i$ a p-value $p_i(\vartheta_i)$ for testing $H^{(\vartheta_i)}_i: \theta_i=\vartheta_i$ against $K^{(\vartheta_i)}_i:\theta_i>\vartheta_i$, 
i.e.\ $\mathbb{P}_{\theta}\big(p_i(\theta_i)\le\alpha\big)\le\alpha$ for all $\alpha\in (0,1)$ and $\theta\in\mathbb{R}^2$. We also assume that $p_i(\theta_i)$ is increasing and continuous in each $\theta_i\in\mathbb {R}$. Similar to Example 1, we define for each (continuous) parameter $\theta_i$ a non-increasing and continuous function $a_i:\theta_i \to [0,\infty)$ to build the weights $w_i(\theta)=a_i(\theta_i)/\sum_{j=1}^2 a_j(\theta_j)$, $i=1,2$. We use these weights
to reject $\theta\in\mathbb{R}^2$ with the p-value $p(\theta)=\min_{i=1}^2 p_i(\theta_i)/w_i(\theta)$ and corresponding decision function $\psi_\theta=\mathbf{1}_{\{p(\theta)\le \alpha\}}$.
We aim to show that $(\psi_\theta)_{\theta\in\Theta}$ satisfy the weak uniform-consonant property and 
how we can verify at a specific point $\theta$ and given direction $i$, whether we have weak uniform-consonance in this direction or not.

Assume (without loss of generality) that $p_1(\theta_1)/w_1(\theta)=\min_{j=1}^2 p_j(\theta_j)/w_j(\theta)\leq\alpha$.
Weak uniform-consonance in direction $i=1$ follows because
$$
p_1(\theta_1)/w_1(\theta)
=\big(a_1(\theta_1)+a_2(\theta_2)\big)\, p_1(\theta_1)/a_1(\theta_1)
=\big(1+a_2(\theta_2)/a_1(\theta_1)\big)\, p_1(\theta_1)
$$
is increasing in $\theta_1$ and non-increasing in $\theta_2$. Therefore, $p(\theta)=p_1(\theta_1)/w_1(\theta)\le\alpha$  implies $p(\theta')\le p_1(\theta'_1)/w_1(\theta')\le \alpha$ for all $\theta'_1\le \theta_1$
and $\theta'_2\ge \theta_2$.
Of course, if $p_2(\theta_2)/w_2(\theta)=p_1(\theta_1)/w_1(\theta)$, we have weak uniform-consonance also in direction $j=2$.

Assume now that 
\begin{align}\label{eq:cwH_case2}
p_2(\theta_2)/w_2(\theta)>p_1(\theta_1)/w_1(\theta)
\end{align}
We show below that we have weak uniform-consonance in direction $i=2$ if and only if $p(\tilde{\theta}_1,\theta_2)\le\alpha$, where %$\tilde{\theta}_2=\theta_2$, and 
$\tilde{\theta}_1$ 
is such that $p_1(\tilde{\theta}_1)/a_1(\tilde{\theta}_1)=p_2(\theta_2)/a_2(\theta_2)$. Note that 
by our assumptions, $p_1(\vartheta_1)/a_1(\vartheta_1)$ is increasing and continuous in $\vartheta_1$ and therefore 
$\tilde{\theta}_1> \theta_1$ can be determined  by a standard root finding procedure. 

Obviously, $p(\tilde{\theta}_1,\theta_2)>\alpha$ contradicts weak uniform-consonance at $\theta$ in direction $j=2$ (note that $\tilde{\theta}_1>\theta_1$). If $p(\tilde{\theta}_1,\theta_2)\le \alpha$, we can see by similar arguments as above, that
for all $\vartheta_1\ge \tilde{\theta}_1$
$$p(\vartheta_1,\theta_2)\le \big(a_1(\vartheta_1)+a_2(\theta_2)\big)\,p_2(\theta_2)/a_2(\theta_2)\le p(\tilde{\theta}_1,\theta_2)\le \alpha.$$ 
Moreover, for 
$\vartheta_1\le \tilde{\theta}_1$ we have that 
$p(\vartheta_1,\theta_2)=\big(1+a_2(\theta_2)/a_1(\vartheta_1)\big) p_1(\vartheta_1)$ is increasing in $\vartheta_1$ and therefore $p(\vartheta_1,\theta_2)\le p(\tilde{\theta}_1,\theta_2) \le\alpha$ for all 
$\vartheta_1\le \tilde{\theta}_1$. In summary, we have shown that, for a given $\theta_2$ with \eqref{eq:cwH_case2}, the maximum of $p(\vartheta_1,\theta_2)$ for $\vartheta_1\ge \theta_1$ is attained at the $\vartheta_1=\tilde{\theta}_1$ for which $p_1(\tilde{\theta}_1)/a_1(\tilde{\theta}_1)=p_2(\theta_2)/a_2(\theta_2)$. Of course, this means that $p(\tilde{\theta}_1,\theta_2) \le\alpha$ implies weak consonance in direction $i$, but is not yet known to imply weak uniform-consonance in this direction.

To also verify this, we need to show that $p(\vartheta_1,\theta'_2)\le p(\tilde{\theta}_1,\theta_2)$ for all $\vartheta_1\ge \theta_1$ and $\theta'_2\le \theta_2$. To this end, note that, by the previous arguments, for all $\vartheta_1\ge \theta_1$ and $\theta'_2\le \theta_2$ we get $p(\vartheta_1,\theta'_2)\le p(\tilde{\theta}'_1,\theta'_2)$ where $\tilde{\theta}'_1$ satisfies $p_1(\tilde{\theta}'_1)/a_1(\tilde{\theta}'_1)=p_2(\theta'_2)/a_2(\theta'_2)$. Now, the latter identity implies
$$p(\tilde{\theta}'_1,\theta'_2)=a_1(\tilde{\theta}'_1) \frac{p_1(\tilde{\theta}'_1)}{a_1(\tilde{\theta}'_1)} + a_2(\theta'_2) \frac{p_2(\theta'_2)}{a_2(\theta'_2)}=   p_1(\tilde{\theta}'_1) +p_2 (\theta'_2)\le  p_1(\tilde{\theta}_1) +p_2 (\theta_2)=p(\tilde{\theta}_1,\theta_2)$$
and therefore $p(\vartheta_1,\theta'_2)\le p(\tilde{\theta}_1,\theta_2)$ for all $\vartheta_1\ge \theta_1$ and $\theta'_2\le \theta_2$. This shows that $p(\tilde{\theta}_1,\theta_2)\le\alpha$  implies even weak uniform-consonance.

We call the introduced procedure a `continuously weighted Holm' procedure since it generalizes the Bonferroni-Holm procedure for multiple tests (i.e.\ binary parameters) to continuous parameters with continuous weights.   
In the Appendix we extend this procedure to the general case of $m\ge 2$ parameter $\theta_i$, where each parameter point $\theta\in\mathbb{R}^m$ is tested by a weighted Bonferroni test with weights $w_i(\theta)=a_i(\theta_i)/\sum_{j=1}^m a_j(\theta_j)$ based on individual, non-increasing and continuous functions $a_i:\vartheta_i\in\mathbb{R}\to [0,\infty)$. In this case, the verification (or falsification) of weak uniform consonance in a specific direction $j$ requires the determination of up to $m-1$ parameter values $\tilde{\theta}_l$ that are similar to the above $\tilde{\theta}_1$. 

\subsection{Projection algorithm for continuous parameter}

We introduce now an algorithm for the general situation of $m$ continuous parameters $\theta_i\in\Theta_i$ (an interval) that is based on weak uniform-consonance and the possibility to verify this at each point $\theta$ and direction~$i$. As in the discrete case, the algorithm requires an initial lower bound $\theta_0=(\theta_{10},\ldots,\theta_{m0})\le L$ for the (unknown) confidence bound $L$. If all $\Theta_i$ are bounded from below, the algorithm can start at the origin. Otherwise, we need (as in the discrete parameter case) an initialization algorithm that provides data-driven lower bounds for all (from below unbounded) $\theta_i$.  In the following, we will assume the existence of such an initialization algorithm. Later, we will present an example of an initialization algorithm for a concrete situation.

A first, but rough projection algorithm could be to discretize the parameter spaces $\Theta_i$ and apply Algorithm 1 on the resulting grid. However, sufficiently good approximations will (in general) require rather fine grids and result in computationally intensive algorithms. We therefore suggest below the more refined Algorithm 2 that consist of passing through a rough grid and then continuing with a bisection search to approximate $L_i$ up to the required precision $\varepsilon$.

\begin{algorithm}[H]
\setstretch{1.1}
\caption{Projection algorithm for continuous parameter with discretization} \label{alg_abstract_continuous}
\stepone{
Initialize the lower bounds to $\lambda=(\theta_{10},\dots,\theta_{m0})$ where $\lambda_i=\theta_{i0}\le L_i$ for all $i=1,\ldots,m$\;
Define a discrete grid starting at $\lambda$ with ordered points $\theta_{i0}<\theta_{i1}<\theta_{i2},\ldots$, $1\leq i\leq m$\;
Set $k_i\rightarrow 0$ for all $1\leq i\leq m$\;
}
\stepgrid{
\For{each component $i=1,\ldots,m$}{
\While{weak uniform-consonance at $\lambda$ in direction $i$}{
step to the next parameter point in the $i$-th coordinate direction, i.e.:\\
update $\lambda_i\rightarrow\theta_{i,k_i+1}$ and $k_i\rightarrow k_i+1$\;
}
Set upper approximation: $\nu_i\rightarrow \lambda_i$\;
Set $\lambda_i\rightarrow \theta_{i, k_i-1}$\;
}
}
\setstretch{1.1}
\stepbisec{
\For{each component $i=1,\ldots,m$}{
Initialize current grid point $\xi \rightarrow \lambda$\;
\While{$\lambda_i-\nu_i>\varepsilon$ %$\text{iterations}<\text{maxIterations}$
}{
Update current grid point $\xi$: Set $\xi_i=(\lambda_i+\nu_i)/2$ and fix other components\;
\eIf{weak uniform-consonance at $\xi$  in direction $i$}{
Update lower approximation $\lambda_i\rightarrow(\lambda_i+\nu_i)/2$\;
}{
Update upper approximation $\nu_i\rightarrow(\lambda_i+\nu_i)/2$\;
}
}
}
}
\Return{Lower approximation $\lambda$ and upper approximation $\nu$ of $L$.}
\end{algorithm}

Algorithm~\ref{alg_abstract_continuous} is illustrated in Figure~\ref{alg:projContinParam} for two parameters. As illustrated in this figure and shown in the Appendix for general $m$, we obtain $\varphi^{\text{proj}}_i(\lambda_i)=1$ (or $\varphi^{\text{proj}}_i(\xi_i)=1$) at every iterative step with current test parameter value $\lambda$ (or $\xi$) if we can show weak uniform-consonance at the current $\lambda$ (or $\xi$) in the direction $i$ along we move. This implies that the $i$-th component of the current constellation provides a lower bound for $L_i$. Obviously, if we fail to show weak uniform-consonance, then $\varphi^{\text{proj}}_i(\lambda_i)=0$ (or $\varphi^{\text{proj}}_i(\xi_i)=0$), and the current constellation provides an upper bound for $L_i$. 

\begin{figure}[h]
 \begin{minipage}{0.45\textwidth}
    \centering
     \begin{tikzpicture}[scale=0.7]
  \draw[->] (0,0) -- (5.5,0) node[right] {$\theta_1$};
  \draw[->] (0,0) -- (0,5) node[above] {$\theta_2$};
\def\circles{
  0/0, 0/1, 0/2, 0/3, 0/4,
  1/0,
  2/1, 2/2, 2/3, 2/4,
  3/1, 3/2, 3/3, 3/4,
  4/1, 4/2, 4/3, 4/4,
  5/1, 5/2, 5/3, 5/4
}
\foreach \x/\y in \circles {
    \fill (\x,\y) circle (2pt) {};
}

\def\rectangles{
  0/0, 0/1, 0/2, 0/3, 0/4,
  1/0, 1/1, 1/2, 1/3, 1/4,
  2/0, 2/1,
  3/0, 3/1,
  4/0, 4/1,
  5/0, 5/1
}
\foreach \x/\y in \rectangles {
    \draw (\x,\y) node[draw=red, rectangle, fill=red, minimum size=4pt, inner sep=0pt]{};
}

\draw (1,0) node[draw=red, rectangle, fill=red, minimum size=7pt, inner sep=0pt]{};
\draw (0,0) node[draw=red, rectangle, fill=red, minimum size=7pt, inner sep=0pt]{};

\fill[pattern=north east lines, pattern color=red] (0,0) rectangle (1,4.5);

\draw[decorate, decoration={brace, amplitude=6pt}]
        (3,-0.15) -- (2,-0.15)
        node[midway, below=7pt] {$\varepsilon$};

\draw[->, red, thick] (1,-0.3) -- (0.5,-0.3) {};
\draw[->, red, thick] (1.15,0.25) -- (1.15,0.75) {};

\fill[blue!80]
  plot[smooth cycle] coordinates {
    (1.35,1.35) (2.5,4.5) (5.5, 4.5) (4,2.8) 
  };
\def\circles{
  2/2, 2/3, 2/4,
  3/2, 3/3, 3/4,
 4/2, 4/3, 4/4,
5/2, 5/3, 5/4
}
\foreach \x/\y in \circles {
    \fill (\x,\y) circle (2pt) {};
}

\end{tikzpicture}\\
\text{Step 2 for $i=1$}
\end{minipage}
 \textbf{$\Rightarrow$}   % \hfill
    \begin{minipage}{0.45\textwidth}
    \centering
     \begin{tikzpicture}[scale=0.7]
  \draw[->] (0,0) -- (5.5,0) node[right] {$\theta_1$};
  \draw[->] (0,0) -- (0,5) node[above] {$\theta_2$};

\def\rectangles{
  0/0, 0/1, 0/2, 0/3, 0/4,
  1/0, 1/1, 1/2, 1/3, 1/4,
  2/0, 2/1,
  3/0, 3/1,
  4/0, 4/1,
  5/0, 5/1
}
\foreach \x/\y in \rectangles {
    \draw (\x,\y) node[draw=red, rectangle, fill=red, minimum size=4pt, inner sep=0pt]{};
}

\draw (1,1) node[draw=red, rectangle, fill=red, minimum size=7pt, inner sep=0pt]{};

\fill[pattern=north west lines, pattern color=red] (1,1) rectangle (5.25,0);
\fill[pattern=north west lines, pattern color=red] (0,1) rectangle (5.25,0);
\fill[pattern=north west lines, pattern color=red] (1,0) rectangle (5.25,0);
\fill[pattern=north west lines, pattern color=red] (0,0) rectangle (5.25,0);

\fill[blue!80]
  plot[smooth cycle] coordinates {
    (1.35,1.35) (2.5,4.5) (5.5, 4.5) (4,2.8) 
  };
\def\circles{
  2/2, 2/3, 2/4,
  3/2, 3/3, 3/4,
 4/2, 4/3, 4/4,
5/2, 5/3, 5/4
}
\foreach \x/\y in \circles {
    \fill (\x,\y) circle (2pt) {};
}

\draw[decorate, decoration={brace, amplitude=6pt}]
        (3,-0.15) -- (2,-0.15)
        node[midway, below=7pt] {$\varepsilon$};

\fill[pattern=north east lines, pattern color=red] (0,0) rectangle (1,4.5);
\draw[->, red, thick] (1.25,1.15) -- (1.75,1.15) {};
\draw[->, red, thick] (1.25,0.75) -- (1.25,0.25)  {};

\end{tikzpicture}\\
\text{Step 2 for $i=2$}
\end{minipage}
\caption{Grid Traversal (Step 2) of projection algorithm for continuous parameter}\label{alg:projContinParam}
\end{figure}
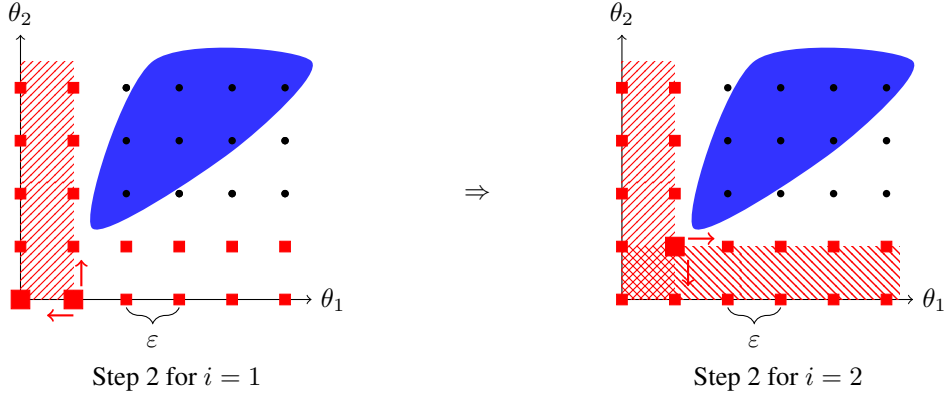

We obtain a conservative approximation, if we report the final lower bound $\lambda$ of Algorithm 3, and should also report the final precision $\lambda-\nu$. The latter is particularly important if, for practical reasons, we impose in step 3 a restriction on the number of steps in the bisection search and stop when this number is reached before reaching the desired precision.   

We end this section, recalling that for Algorithm 2, we need to be able to verify at each (grid) point whether weak uniform-consonance applies in a given direction $i$.  While an (easy) verification algorithm is available for the continuously weighted Holm procedure of the previous section (see also the appendix), we are not aware of such an algorithm for the more general class of procedures introduced in the next section. This class will consist of monotonously weighted Bonferroni tests for continuous parameter that generalizes the weighted Bonferroni test of \citet{hommel2007powerful} and the informative simultaneous SCI of \citet{brannath2014new}, \citet{schmidt2014informative, schmidt2015informative}, \citet{informativeSCIBrannathKlugeScharpenberg}, and \citet{klugeBrannath2026GSD}. The new algorithm will also be based on weak uniform-consonance.
\section{Monotonously weighted Bonferroni tests for continuous parameter}\label{sec:monWeightBonf}

In this section, we consider the partitioning principle with a general class of local tests $(\psi_\theta)_{\theta\in\mathbb{R}^m}$,
%for each $\theta\in\Theta=\times_{i=1}^m\Theta_i$ 
that generalizes the more specific classes in \citet{brannath2014new}, \citet{schmidt2014informative, schmidt2015informative}, \citet{informativeSCIBrannathKlugeScharpenberg}, and \citet{klugeBrannath2026GSD}.
We assume, like in the mentioned literature, that for each $i\in I=\{1,\ldots,m\}$ and parameter value $\vartheta_i\in\Theta_i=\mathbb{R}$ a marginal p-value $0<p_i(\vartheta_i)<1$  
for testing $H^{(\vartheta_i)}_i: \theta_i=\vartheta_i$ exists, that satisfies 
$\mathbb{P}_\theta\big(p_i(\theta_i)\le\alpha\big)\le\alpha$ for all $\alpha\in (0,1)$ and $\theta\in\Theta$. We also assume that these p-values have the following additional properties:  
   \begin{itemize}
        \item[(i)] $p_i(\vartheta_i)$ is increasing and continuous in each $\vartheta_i\in\mathbb {R}$,
        \item[(ii)] $\lim_{\vartheta_i\rightarrow -\infty} p_i(\vartheta_i) = 0$ 
        %\item[(iii)] 
        and $\lim_{\vartheta_i\rightarrow\infty} p_i(\vartheta_i)>\alpha$. 
        % In particular, $\left(p_i(\theta_i)=0\right)\Leftrightarrow\left(\theta_i=-\infty\right)$
    \end{itemize}
Property (i) typically follows when using one-sided p-values for testing $H^{(\vartheta_i)}_i: \theta_i=\vartheta_i$ against $K^{(\vartheta_i)}_i:\theta_i>\vartheta_i$. This and the other properties are satisfied for many commonly used (asymptotic) Gauss- or t-tests.

Like in the discrete case, we assume for all $i\in I$ positive weights $w_i(\theta)$, $\theta\in \mathbb{R}^m$, 
that depend now on $\theta$ continuously, are non-increasing in $\theta_i$ and non-decreasing in all $\theta_j$
for $j\not= i$. Given these weights, we use the decision function $\psi_\theta=\mathbf{1}_{\{p(\theta)\le \alpha\}}$ with  $p(\theta)=\min_{i=1}^m p_i(\theta_i)/w_i(\theta)$.
Since $p_i(\theta_i)/w_i(\theta)$ is increasing in $\theta_i$ and non-increasing in $\theta_j$ for all $j\not=i$, we obtain for the component $i$ which determines the minimum:
$$p(\theta)=p_i(\theta_i)/w_i(\theta) \le \alpha\ \Rightarrow\ 
p(\theta')\le p_i(\theta'_i)/w_i(\theta') \le \alpha \text{ for all }
\theta' \text{ with }\theta'_i\le \theta_i\text{ and }\theta'_j
\ge \theta_j\text{ for all  }j\not=i.$$    
This implies weak uniform-consonance of 
$(\psi_\theta=\mathbf{1}_{\{p(\theta)\le \alpha\}})_{\theta\in\mathbb{R}^m}$. Remember that our goal is to obtain the projection of $C:=\{\theta\in\mathbb{R}^m : p(\theta)>\alpha\}$. 

We do not need property (ii) and the continuity of $p_i(\vartheta_i)$ and $w_i(\theta)$ for weak uniform-consonance. We need these properties for the below presented iterative Algorithm 3 to provide lower approximations for the confidence bounds, $\lambda\le L$, that converge to $L$. As originally suggested in \citet{klugeBrannath2026GSD}, the algorithm also provides upper approximations $\nu\ge L$ 
that converge to $L$, which permits to stop the algorithm with a pre-specified precision $\varepsilon$, namely when $ \max_{i=1}^m |\lambda_i-\nu_i|\le \varepsilon$.
This generalizes and improves the algorithms suggested in \citet{brannath2014new}, \citet{schmidt2014informative, schmidt2015informative} and \citet{informativeSCIBrannathKlugeScharpenberg}.

The above-mentioned convergence and other properties of the lower and upper approximations in Algorithm 3 are verified in the Appendix. 
As mentioned for the previously presented algorithms, we may stop it with a maximum number of iterations before the precision $\varepsilon$ is reached. Reporting the finally received precisions $|\lambda_i - \nu_i|$, $i=1,\ldots,m$, or their maximum 
is then essential. In any case, the final $\lambda$ provides a conservative lower approximation of $L$ and the resulting SCI has coverage probability of at least $1-\alpha$.

To understand the role of the weak uniform-consonance property for Algorithm 3, note that in each iteration step, for the constellation $\tilde{\lambda}^{(i)}=(\lambda_1,\dots,\lambda_{i-1},\lambda^{\text{new}}_i,\lambda_{i+1}\dots,\lambda_m)$,
weak uniform-consonance in direction $i$ is obviously satisfied. Hence, instead of moving to the next grid or bisection point of Algorithm 2, the algorithm moves along each coordinate $i$ to the next point where weak uniform-consonance in this direction is still obvious. This implies $\tilde{\lambda}^{(i)}\le L$.
The convergence of the lower approximations must then be verified by additional arguments that are specific to the class of local tests under consideration and is addressed in the Appendix. 

\begin{algorithm}[H]
\caption{Algorithm for monotonously weighted Bonferroni tests with continuous parameter}\label{algorithm:Hommel}
\init{Fix thresholds $\varepsilon>0$, $1>\delta_0>0$ (typically small) and initialize $k\rightarrow 1$ and $\delta\to \delta_0$,
as well as the approximations, $\lambda \to \lambda_0$ and $\nu\to\nu_0$, with initial bounds $\lambda_0\leq L\le \nu_0\in C$ that satisfy
\begin{align}\label{eq:initA3}
p_{i}(\lambda_{i,0})\leq w_i(\lambda_0)\alpha \quad \text{and}\quad p_{i}(\nu_{i,0})\geq w_i(\nu_0)(\alpha+\delta_0)\quad \text{for all } i=1,\dots,m. \;
\end{align}
}
\iter{
\While{$\max_{i=1}^m |\lambda_i - \nu_i| \geq \varepsilon$}{
\For{each component $i=1,\dots,m$}{
Calculate $\lambda^{\text{new}}_i\geq\lambda_i$ such that
$$ p_{i}(\lambda^{\text{new}}_i)=w_i(\lambda_1,\dots,\lambda_{i-1},\lambda^{\text{new}}_i,\lambda_{i+1}\dots,\lambda_m)\alpha .\;$$

Calculate $\nu^{\text{new}}_i\leq\nu_i$ such that
$$ p_{i}(\nu^{\text{new}}_i)=w_i(\nu_1,\dots,\nu_{i-1},\nu_i^{\text{new}},\nu_{i+1},\dots,\nu_m)(\alpha+\delta). \;$$

Update $k\rightarrow k+1$, (increase) $\lambda_i\rightarrow\lambda^{\text{new}}_i$, (decrease) $\nu_i\rightarrow\nu_i^{\text{new}}$ and (decrease) $\delta \rightarrow \delta_0^k$\;
}
}}
\Return{Lower and upper approximation $\lambda$ and $\nu$ and accuracy $\max_{i=1}^m |\lambda_i - \nu_i|$.}
\end{algorithm}

We end this section with an example on how to achieve the initial lower and upper bounds $\lambda_0$ and $\nu_0$. %We are not aware of a general strategy for obtaining an initial lower bound, however, 
Choosing $\delta_0<1-\alpha$, the initial lower bound $\nu_0$ can be easily determined by the component-wise non-adjusted confidence bound at the level $\alpha+\delta_0$, i.e.\ for all $i=1,\ldots,m$ the 
$\nu_{0,i}$ satisfying $p_i(\nu_{0,i})=\alpha+\delta_0$.
As in the discrete case, a lower bound can be easily determined under the assumption that there exists for all components $i$ a $\theta_i^{*}\in\mathbb{R}$ such that $w_i(\theta)\geq w_{i0}>0$ for all $\theta\in\Theta$ with $\theta_i\leq\theta_i^*$. The initial value $\lambda_{i,0}$ can then be determined as the minimum of $\theta_i^*$ and the lower one-sided confidence bound at level $w_{i0}\alpha$, i.e. 
\begin{align*}
    \lambda_{i,0}:=\min\{\theta_i^*, p_i^{-1}(w_{i0}\alpha)\}
\end{align*}
where $p_i^{-1}(w_{i0}\alpha)$ is the unique solution of $p_i(\theta'_{i})=w_{i0}\alpha$.
This is particularly the case for the weighted monotone Bonferroni tests suggested in \citet{informativeSCIBrannathKlugeScharpenberg} for obtaining informative simultaneous SCI that are close to a given graphical test procedure.

We argue that $\lambda_0$ is indeed a valid starting vector fulfilling the requirements from Algorithm~\ref{algorithm:Hommel}. For this, it must be shown that for all components $i$ and all $\theta'\in\Theta$ with $\theta'_i\leq\lambda_{i,0}$ the parameter $\theta'$ is not contained in the confidence set $C$, i.e. $p(\theta')=\min_{i=1}^m p_i(\theta_i')/w_i(\theta')\leq\alpha$. This is fulfilled if for all $\theta'$ with $\theta_i'\leq\lambda_{i,0}$ we have $p_i(\theta'_i)\leq w_i(\theta')\alpha$. The latter follows from:
$$p_i(\theta_i')\leq p_i(\lambda_{i,0})\leq p_i(p_i^{-1}(w_{i0}\alpha))= w_{i0}\alpha\leq w_i(\theta')\alpha$$
where the last inequality follows from the fact that $\theta_i'\leq\lambda_{i,0}\leq\theta_i^*$ and our assumption on the weights.

\subsection{Algorithm for situations with continuous and discrete parameters}

In some applications, we may have a mixture of discrete and continuous parameters, for example, when focusing on hypothesis tests in some of the parameters (e.g.\ reducing potentially continuous parameters to binary) and aiming for informative confidence intervals for the other ones. In this case, one could include the discrete components in Algorithm \ref{alg_abstract_continuous} with its values as grid points, whereby it is sufficient to verify weak consonance in direction $j$ for each discrete component $j$ and parameter $\lambda\in\Theta$. For the continuous components we still need to verify the stronger weak uniform-consonant property. The \emph{Grid Traversal}-part of Algorithm \ref{alg_abstract_continuous} is then performed for all components until for all discrete components the confidence bounds are determined exactly. If for a continuous component the approximation accuracy $\varepsilon>0$ has then not been reached, the \emph{Bisection Search} can be performed only for the continuous components.

How a combination of Algorithm \ref{alg_abstract_discrete} and Algorithm \ref{algorithm:Hommel} could look like remains an open question. As for the use of Algorithm \ref{algorithm:Hommel} for continuous parameters alone, one could move along each coordinate $i$ to the next point where weak uniform-consonance (continuous parameter) or weak consonance (discrete parameter) in this direction is still obvious. The derivation of sufficient conditions such that it is stepped far enough and the approximations converge from below against the confidence bounds remain an unresolved issue.

\section{Summary and discussion}
 
Inspired by the computational equivalence between the closed testing and partitioning principle, we have extended the concept of consonance for closed tests to consonance properties for the partitioning principle. The goal of this  extension is to obtain efficient and feasible algorithms for the implementation of the partitioning principle. As a running example we have considered the extension of the consonant weighted Bonferroni closed tests in \citet{hommel2007powerful} to the partitioning principle with discrete and continuous parameter. Our work improves and sheds light on algorithms suggested in \citet{brannath2014new}, \citet{schmidt2014informative, schmidt2015informative}, \citet{informativeSCIBrannathKlugeScharpenberg}, and \citet{klugeBrannath2026GSD} for obtaining informative simultaneous confidence intervals.

It is important to note that the concept of weak (uniform) consonance introduced in this work does not immediately imply the rejection of a marginal hypothesis, and therefore is not just an extension of the consonance concept for closed tests. Weak (uniform) consonance is weaker than (uniform) consonance and thereby easier to achieve and verify, which is helpful for applications. Based on the weak consonance property, we have provided efficient computational algorithms for discrete parameters, and for continuous parameters, feasible as well as efficient algorithms under the stronger weak uniform consonance property. The latter provide an conservative and anti-conservative approximation for each lower confidence bound with a pre-defined precision. 

In the continuous parameter case, weak \textit{uniform}-consonance is required to be able to extend given rejections on grid points to the continuum between them. In contrast to the discrete case, the presented algorithms also need to move along already rejected regions and therefore require a method to conclude weak uniformly consonance in a specific direction. The determination of weakly uniform consonance in a given direction is rather easy for an extension of Holm's procedure to continuous parameter, but it seems unfeasible for the general class of weighted Bonferroni tests. Therefore, we have suggested an alternative algorithm for the latter class of local tests, extending and improving algorithms suggested in \citet{brannath2014new}, \citet{schmidt2014informative, schmidt2015informative} and \citet{informativeSCIBrannathKlugeScharpenberg}.

In \citet{informativeSCIBrannathKlugeScharpenberg} and \citet{klugeBrannath2026GSD} graphical tests with gatekeepers were also considered, where some parameters (e.g.\ for secondary endpoints) are tested only after rejection of specific null hypotheses for other parameters (e.g.\ for primary endpoints). This is a situation that is related to the above-discussed cases with discrete and continuous components of the parameter vector, but is more complex because it includes cases in which the original and a discretized version of the same parameter are considered, violating the variational independence assumptions made throughout this paper. Including gatekeeping requires modifications of Algorithm 2 and 3. For the specific case of informative SCI for graphical tests, algorithms similar to Algorithm 3 (without upper approximations) that can cope with gatekeeping are given in \citet{informativeSCIBrannathKlugeScharpenberg} and extended to graphical group sequential tests (with upper approximations) in \citet{klugeBrannath2026GSD}. The development of similar algorithms for the more general continuously weighted Bonferroni tests with gatekeeping is yet an unresolved issue.      

A further open research question is the derivation of efficient algorithms for restricted parameters, like the ones underlying all pair-wise comparisons of multiple treatment groups. All-pairwise comparisons have recently been suggested for clinical trials without a control group (\citealp{Burnett2026}). Another open research issue is the derivation of general and efficient algorithms for the implementation of the partitioning principle with level exhaustive local tests that  account for the joint distribution of the underlying test statistics. This is an even more challenging task as the weak (uniform)-consonance property is easily violated when using individual weights that do not only depend on the corresponding but also other parameter components (see e.g.\ \citealp{brannath2014new}).  

We end this work with a general discussion on the utility of simultaneous confidence intervals (SCI) for clinical trials. We consider them as important element of frequentists analyses with multiple confirmatory goals. In general, confidence intervals are more informative than hypothesis tests, and they do not suffer from well-known issues with p-values that can be small even under irrelevant treatment effects.  One could even go as far as to conclude that hypothesis tests and p-values are dispensable when providing (simultaneous) confidence intervals, because claims on the targeted null hypotheses can easily be read off (simultaneous) confidence intervals (and even more, if sufficiently informative). When asking for multiplicity corrections with multiple tests, asking for simultaneous coverage probabilities for multiple parameters is very natural, in particular, because -- as we have seen -- FWER control is mathematically equivalent to simultaneous coverage of the corresponding (less informative) binary indicator variables. 
Like for the null hypothesis, accounting for different preferences for different parameter and different parameter values can be a valuable task for clinical trials. For example, achieving more power for excluding small (and thereby less realistic) parameter values for the price of reducing power for larger ones (that are anyhow less likely to be excluded), appears to be a reasonable strategy. This could be the motivation for the use of the above-described continuously weighted Holm tests and corresponding SCI based on the partitioning principle and related projection algorithm. The informative SCI suggested in  \citet{brannath2014new}, \citet{schmidt2014informative, schmidt2015informative}, \citet{informativeSCIBrannathKlugeScharpenberg} and \citet{klugeBrannath2026GSD}, that builds on graphical tests, are motivated by the intention to account for such and other preferences in the construction of SCI. These methods could be considered as a starting point for a more interval (and less binary) based frequentist analysis of clinical data to be extended by methodologies that account for preference independently of initial null hypotheses, e.g.\ via optimality considerations based on parameter dependent gains and losses. We believe that this is a highly valuable topic for future research.

\section*{Acknowledgement}
The authors gratefully acknowledge the support of the Leibniz ScienceCampus Bremen Digital Public Health (www.digital-public-health.de), which is jointly funded by the Leibniz Association (W72/2022), the Federal State of Bremen, and the Leibniz Institute for Prevention Research and Epidemiology – BIPS.

\bibliographystyle{apalike}
\bibliography{references} 

\appendix
\renewcommand{\thesection}{\Alph{section}}

\titleformat{\section}
{\normalfont\large\bfseries}
{Appendix \thesection:}
{0.5em}
{}

\section{Mathematical results for Algorithm~\ref{alg_abstract_discrete}}
\begin{lemma}\label{lem:LemmaForAlgoDiscrete} 
In each step of Algorithm~\ref{alg_abstract_discrete} with $\psi_{\lambda}=1$ we have $\lambda\le L$, i.e.\  $\lambda_j\leq L_j$ for all $j=1,\dots,m$.
\end{lemma}
\begin{proof}

    We show the statement by induction in algorithm's steps. By our assumptions in Section~\ref{sec:WCandAfDP}, we have $\lambda\leq L$ (component-wise) for the initial $\lambda=(\theta_{10},\dots,\theta_{m0})$ . Making the induction assumption that at the current  step we have $\lambda_j=\theta_{j k_j}\le L_j$ for all $j=1,\ldots,m$, the algorithm only continues with the next step if, we find $i$ such that $\psi_{\lambda}=1$ 
     for $\lambda_i=\theta_{i k_i}$ and the area $G_{\lambda}^{(i)}$ is not included in $C$. This implies $\theta_{i k_i}<L_i$ and that the update in the $i$-th component fulfils $\theta_{i k_i+1}\leq L_i$. Because the other components remain unchanged, we have for the updated $\lambda$ that $\lambda\leq L$ component-wise.
\end{proof}

\section{Continuously weighted Holm procedure}

We assume $m$ parameter $\theta_i\in \mathbb{R}$ and for each parameter, a continuous and non-increasing function 
$a_i:\vartheta_i\in\mathbb{R}\to (0,\infty)$. We also consider weighted Bonferroni tests with continuous weights $w_i(\theta)=a_i(\theta_i)/\sum_{j=1}^m a_j(\theta_j)$, and individual p-values $p_i(\theta_i)$ that are continuous and increasing, leading to the local p-values $p(\theta)=\min_{i=1}^m p_i(\theta_i)/w_i(\theta)$. Note that by these assumptions each $p_j(\tilde{\theta}_j)/a_j(\tilde{\theta}_j)$ is increasing in $\tilde{\theta}_j$. The following result provides a method to show that    
$(\psi_\theta)_{\theta\in\mathbb{R}^n}$ with
$\psi_\theta
=\mathbf{1}_{\{p(\theta)\le \alpha\}}$ is weakly uniform-consonance at $\theta$ in a given direction $i$.
\begin{theorem} \label{th:cwHolm}
Let $\theta\in\mathbb{R}^n$, $i\in \{1,\ldots,m\}$ and $\tilde{\theta}^{(i)}=(\tilde{\theta}^{(i)}_1,\ldots,\tilde{\theta}^{(i)}_m)$ be defined by $\tilde{\theta}^{(i)}_j=\theta_j$ if 
$p_j(\theta_j)/a_j(\theta_j) > p_i(\theta_i)/a_i(\theta_i)$
and defined as the solution of
$$p_j(\tilde{\theta}_j^{(i)})/a_j(\tilde{\theta}_j^{(i)})=p_i(\theta_i)/a_i(\theta_i)\quad\text{if}\quad  
p_j(\theta_j)/a_j(\theta_j) \le p_i(\theta_i)/a_i(\theta_i).
$$
The above introduced family of local tests $(\psi_\theta)_{\theta\in\mathbb{R}^n}$ is weakly uniform-consonance at $\theta$ in direction $i$ if and only if $p(\tilde{\theta}^{(i)})\le  \alpha$.
\end{theorem}
We prove the theorem with following two lemmas. 
\begin{lemma}\label{lem:a1}
For $J\subseteq\{1,\ldots,m\}$ let $\bar{J}:=\{1,\ldots,m\}\setminus J$ and 
$$M_J:=\Big\{\theta=(\theta_1,\ldots,\theta_m)\in \mathbb{R}^m:  \
\min_{i=1}^m \frac{p_i(\theta_i)}{a_i(\theta_i)}=
\min_{l\in J} \frac{p_l(\theta_l)}{a_l(\theta_l)}=\frac{p_j(\theta_j)}{a_j(\theta_j)} \text{ for all } j\in J \Big\}.$$  
If $\theta,\tilde{\theta}\in M_J$ with $\theta_i=\tilde{\theta}_i$ for all $i\in \bar{J}$, then 
$p(\theta)>p(\tilde{\theta})$ if and only if $\theta_j>\tilde{\theta}_j$ for at least one $j\in J$. Moreover, the latter implies  
$\theta_l>\tilde{\theta}_l$ for all $l\in J$.

\end{lemma}
\begin{proof} For $\theta\in M_J$ and arbitrary $k\in J$ we get   
$$ p(\theta)=\sum_{i=1}^m a_i(\theta_i) \min_{l\in J} \frac{p_l(\theta_l)}{a_l(\theta_l)} 
=\sum_{i\in\bar{J}} a_i(\theta_i) \frac{p_k(\theta_k)}{a_k(\theta_k)} + \sum_{j\in J} p_j(\theta_j),$$
which is increasing in each $\theta_j$, $j\in J$. Hence, $p(\theta)>p(\tilde{\theta})$ for $\theta,\tilde{\theta}\in M_J$ with $\theta_i=\tilde{\theta}_i$ for all $i\in \bar{J}$ implies $\theta_j>\tilde{\theta}_j$ for at least one $j\in J$. 

On the other hand, $\theta,\tilde{\theta}\in M_J$ with $\theta_i=\tilde{\theta}_i$ for all $i\in \bar{J}$ and $\theta_l>\tilde{\theta}_l$ for all $l\in J$ implies $p(\theta)>p(\tilde{\theta})$.
Because each $p_l(\theta_l)/a_l(\theta_l)$ is increasing in $\theta_l$, we have that $\theta_j>\tilde{\theta}_j$ for at least one $j\in J$ implies
$$p_l(\theta_l)/a_l(\theta_l)=p_j(\theta_j)/a_j(\theta_j)> 
p_j(\tilde{\theta}_j)/a_j(\tilde{\theta}_j)=p_l(\tilde{\theta}_l)/a_l(\tilde{\theta}_l) \quad \text{ for all }l\in J,$$ 
which in turn implies that $\theta_l>\tilde{\theta}_l$ for all $l\in J$. So,  $\theta_j>\tilde{\theta}_j$ for at least one $j\in J$ implies $p(\theta)>p(\tilde{\theta})$.
\end{proof}

The next lemma directly implies Theorem~\ref{th:cwHolm}.
\begin{lemma}\label{lem:a2}
Let $\theta\in\mathbb{R}^n$, $i\in \{1,\ldots,m\}$ and $\tilde{\theta}^{(i)}$ as in Theorem~\ref{th:cwHolm}. Then 
$p(\tilde{\theta}^{(i)})=\max_{\theta'\in D^{(i)}_\theta} p(\theta')$ for 
$D^{(i)}_\theta:=\{\theta'\in \mathbb{R}^m:\theta'_i\le\theta_i \text{ and } \theta'_j\geq\theta_j \text{ for all }j\not=i\}$.
\end{lemma}
\begin{proof} We choose an arbitrary $\theta\in\mathbb{R}^n$ and renumber the parameter such that
\begin{align}\label{eq:wlog}
p_{1}(\theta_{1})/a_{1}(\theta_{1})\le p_{2}(\theta_2)/a_{2}(\theta_{2})\le \cdots \le 
p_{m}(\theta_{m})/a_{m}(\theta_{m}).
\end{align}
Note that $\tilde{\theta}^{(i)}\in D^{(i)}_\theta$ and therefore it is sufficient to show that $p(\theta')\le p(\tilde{\theta}^{(i)})$ for all $\theta'\in D^{(i)}_\theta$.

So let $\theta'\in D^{(i)}_\theta$ be arbitrary. Note that by \eqref{eq:wlog} and the definition of $D_{\theta}^{(i)}$ we have $p_{i}(\theta'_{i})/a_{i}(\theta'_{i})\le p_{i}(\theta_{i})/a_{i}(\theta_{i})\le p_{l}(\theta_{l})/a_{l}(\theta_{l})\le p_{l}(\theta'_{l})/a_{l}(\theta'_{l})$ for all $l> i$ (while 
$p_{j}(\theta'_{j})/a_{j}(\theta'_{j})$ for $j<i$ can be smaller than $p_{i}(\theta'_{i})/a_{i}(\theta'_{i})$). Therefore, we obtain %for all $\theta'\in D^{(i)}_\theta$:  
\begin{align}\label{eq:iqa1} %\nonumber
p(\theta') & =\sum_{l=1}^m a_l(\theta'_l)\cdot \min_{j\le i} p_j(\theta'_{j})/a_{j}(\theta'_{j}) 
\le \big[\sum_{j\le  i} a_{j}(\theta'_{j}) +\sum_{l=i+1}^m a_{l}(\theta_{l})\big]\cdot \min_{j\le i} p_{j}(\theta'_{j})/a_{j}(\theta'_{j})
\end{align}
whereby the second sum (from $i+1$ to $m$) must be set to zero if $i=m$, and the inequality becomes a equality. 
Hence, if $i<m$, it is sufficient to consider $\theta'\in D^{(i)}_{\theta}$ with $\theta'_l=\theta_l$ for $l> i$.

Let now $h\le i$ such that 
$p_{h}(\theta'_{h})/a_{h}(\theta'_{h})=\min_{j\le i} p_{j}(\theta'_{j})/a_{j}(\theta'_{j})
\le p_{i}(\theta'_{i})/a_{i}(\theta'_{i}).
$
Then \eqref{eq:iqa1} and the monotonicity of $a_l(\theta'_l)$ implies 
$$p(\theta') = 
\big[1+\sum_{j\le   i, j\not=h} a_{j}(\theta'_{j})/a_{h}(\theta'_{h}) +\sum_{l=i+1}^m a_{l}(\theta_{l})/a_{h}(\theta'_{h})\big] 
\cdot p_{h}(\theta'_{h}) 
\le p(\check{\theta'}^{(h)}),
$$
with $\check{\theta'}^{(h)}=(\check{\theta'_1}^{(h)},\ldots,\check{\theta'_m}^{(h)})$ defined by $\check{\theta'_h}^{(h)}=\theta'_{h}$, $\check{\theta'}^{(h)}_l=\theta_l'=\theta_l$ for all $l> i$,
and as the solution of
$$p_j(\check{\theta'_j}^{(h)})/a_j(\check{\theta'_j}^{(h)})=p_{h}(\theta'_{h})/a_{h}(\theta'_{h})\le p_{j}(\theta'_{j})/a_{j}(\theta'_{j})\quad\text{ for all }j\le i, j\neq h$$
where the inequality implies $\check{\theta'}^{(h)}_j\le\theta'_{j}$ for all $j\le i, j\neq h$. Note that for $\check{\theta'}^{(h)}$ we have as for $\theta'$ that $p_{h}(\check{\theta'_h}^{(h)})/a_{h}(\check{\theta'_h}^{(h)})=\min_{j\le i} p_{j}(\check{\theta'_j}^{(h)})/a_{j}(\check{\theta'_j}^{(h)})$.

Since 
$$p_{i}(\tilde{\theta}^{(i)}_{i})/a_{i}(\tilde{\theta}^{(i)}_i) 
= p_{i}(\theta_{i})/a_{i}(\theta_{i}) 
\ge p_{i}(\theta'_{i})/a_{i}(\theta'_{i})
\ge p_{i}(\check{\theta}^{(h)}_{i})/a_{i}(\check{\theta}^{(h)}_i),
$$ 
we obtain from Lemma~\ref{lem:a1}, when applied to $J=\{1,\ldots,i\}$ 
and $\tilde{\theta'}^{(h)},\tilde{\theta}^{(i)}\in M_J$,
that $ p(\check{\theta'}^{(h)})\le  p(\tilde{\theta}^{(i)})$
and therefore $p(\theta')\le  p(\tilde{\theta}^{(i)})$. 
This shows that $p(\theta')\le  p(\tilde{\theta}^{(i)})$ for all $\theta'\in D^{(i)}_\theta$.
\end{proof}

\section{Mathematical results for Algorithm~\ref{alg_abstract_continuous}} 
\begin{lemma}\label{Lemma:Algo2}
 Let $\lambda\in\mathbb{R}^m$ such that $\lambda\leq L$. Assume that all components except one component $i$ are fixed, and we  move along direction $i$ to the next $\lambda_i^{\text{new}}>\lambda_i$, i.e. $\lambda_i^{\text{new}}=\theta_{i k_i+1}$
in the \textsl{Gird Traversal} step or $\lambda_i=\xi_i$ in the \textsl{Bisection Search} step. Further assume, that at  the parameter point $\lambda^{\text{new}}$, where $\lambda_j^\text{new}=\lambda_j$ for $j\neq i$ and $\lambda_i^{\text{new}}>\lambda_i$, we have weak uniform-consonant in direction $i$. Then, $\varphi_i^{\text{proj}}(\lambda^{\text{new}}_i)=1$ and $\lambda^{\text{new}}\leq L$.
In particular, this implies that the approximations $\lambda$ calculated by Algorithm~\ref{alg_abstract_continuous} (Grid Traversal and Bisection Search) provide component-wise lower approximations of the confidence bounds $L$.
\end{lemma}
\begin{proof}
    By assumption, we have $\varphi_j^{\text{proj}}(\lambda_j)=\min_{\theta':\theta_j'\leq\lambda_j}\psi_{\theta'}=1$
    for all components $j=1,\dots,m$. Thus, $\psi_{\theta'}=1$ for all $\theta'\in H(\lambda):=\cup_{j=1}^m\{\theta'\in\Theta:\theta_j'\leq\lambda_j\}$. Hence, to show that $\lambda_i^{\text{new}}\leq L_i$, it remains to guarantee that $\psi_{\theta'}=1$ for all $\theta'\in\{\theta':\theta_i'\leq\lambda_i^{\text{new}}\}\setminus H(\lambda)$. The latter set can be written as $\{\theta'\in\Theta: \lambda_i <\theta_i'\leq\lambda_i^{\text{new}}\text{ and } \theta_j'>\theta_j\text{ for all } j\neq i\}$ and is contained in the set $D_{\lambda^{\text{new}}}^{(i)}$ from (\ref{eq:D_cons}). Therefore, weakly 
    uniform-consonance at $\lambda^{\text{new}}$ in direction $i$ implies $\varphi_i^{\text{proj}}(\lambda_i)=1$ and thereby $\lambda_i^{\text{new}}\leq L_i$. Because all other components $j\neq i$ remained unchanged, we get $\lambda^{\text{new}}\leq L$.
\end{proof} 

\section{Mathematical results for Algorithm~\ref{algorithm:Hommel}}

In this section, we show the convergence of Algorithm~\ref{algorithm:Hommel} and additional properties stated in Theorem~\ref{theoremConvergencHommelInDetail} below. Recall the two key equations, to be solved in each step of Algorithm~\ref{algorithm:Hommel} for the given component $i=1,\dots,m$:
\begin{align}
    p_{i}(\lambda^{\text{new}}_i)=w_i(\lambda_1,\dots,\lambda_{i-1},\lambda^{\text{new}}_i,\lambda_{i+1}\dots,\lambda_m)\alpha \label{eq:HommelAlgoStepLower}
\end{align}
and
\begin{align}
    p_{i}(\nu^{\text{new}}_i)=w_i(\nu_1,\dots,\nu_{i-1},\nu_i^{\text{new}},\nu_{i+1},\dots,\nu_m)(\alpha+\delta).\label{eq:HommelAlgoStepUpper}
\end{align}
One key result of Theorem~\ref{theoremConvergencHommelInDetail}, which is utilized in its proof, is that the vector of lower confidence bounds $L$ equals the unique $\vartheta\in\mathbb{R}^m$ that satisfies
\begin{align}
    p_i(\vartheta_i)=\alpha\cdot w_i(\vartheta), \quad\text{for all $1\leq i\leq m$}~\label{uniquePropertyHommel}.
\end{align}

\begin{theorem}\label{theoremConvergencHommelInDetail}
    Under the assumptions in Section~\ref{sec:monWeightBonf}, the in each step of Algorithm~\ref{algorithm:Hommel} determined  lower and upper approximations, $\lambda$ and $\nu$, for $L$ have the following properties:  
    \begin{itemize}
        \item[(a)] The sequences of $\lambda$ and $\nu$ are component-wise non-decreasing and component-wise non-increasing, respectively, and in each step we get $\lambda\leq L\leq\nu\in C$.
        \item[(b)] The sequences of lower approximations $\lambda$ and upper approximations $\nu$ converge to parameter points $\lambda_{\infty}$ and $\nu_{\infty}$ that both fulfil property (\ref{uniquePropertyHommel}).
        \item[(c)] There exists at most one $\vartheta\in\mathbb{R}^m$ that fulfils property (\ref{uniquePropertyHommel});
        \item[(d)] The approximations $\lambda$ and $\nu$ converge independently of the starting vectors ($\lambda_0$ and $\nu_0$) to $L$.
    \end{itemize}
\end{theorem}
\subsection*{Proof of (a) in Theorem \ref{theoremConvergencHommelInDetail}}
%\begin{proof}
We verify at first that the sequence of approximations $\lambda$ of $L$ exists and is component-wise non-decreasing. 
For this, assume that %for all $j=1,\dots,m$ it holds
\begin{align}
    p_j(\lambda_j)\leq w_j(\lambda)\alpha\quad \text{ for all }j=1,\dots,m.\label{eq:algo3Uneq}
\end{align}
We fix all components of $\lambda$ except one, namely $\lambda_i$. When searching now for a $\lambda_i^{\text{new}}$ such that (\ref{eq:HommelAlgoStepLower}) is met, only the component in which the weight $w_i$ is non-increasing is varied. Since the p-value $p_i(\vartheta_i)$ is increasing in $\vartheta_i\in \mathbb{R}$, the solution $\lambda_i^{\text{new}}$ of (\ref{eq:HommelAlgoStepLower}) must be greater than or equal to $\lambda_i$. The continuity of the p-value and weight, as well as the assumed limit property of the p-value, guarantees the existence and uniqueness of the intersection point $\lambda_i^{\text{new}}\geq \lambda_i$ fulfilling (\ref{eq:HommelAlgoStepLower}).
In particular, we have $\lambda^{\text{new}}\geq \lambda$ component-wise. For all components $j\neq i$ we have $p_j(\lambda_j^{\text{new}})=p_j(\lambda_j)$ because $\lambda_j^{\text{new}}=\lambda_j$; and 
$w_j(\lambda)\alpha\leq w_j(\lambda^{\text{new}})\alpha$ because the weight $w_j$ is non-decreasing in the i-th component. Thus, for the updated $\lambda\rightarrow\lambda^{\text{new}}$ we have (\ref{eq:algo3Uneq}) 
for all components (including $i$). Because the start vector $\lambda_0$ fulfils (\ref{eq:algo3Uneq}), it follows by induction that the sequence of lower bounds $\lambda$ is non-decreasing. Note that we do not need the assumption
$\lambda_0\le L$ for showing that the sequence $\lambda$ is non-decreasing.

It remains to argue that in each step we have $\lambda\leq L$. For the start vector $\lambda_0$ this follows from the assumptions on the initialization algorithm. For all further steps, the inequality $\lambda\leq L$ follows from the weak consonance property argued in Section~\ref{sec:monWeightBonf} and the same arguments as in the proof of Lemma~\ref{Lemma:Algo2}. 

We now prove part (a) for the sequence of upper approximations $\nu$. For this, assume that for the current $\delta\in (0,\delta_0)$ we have
\begin{align}
    p_j(\nu_j)\geq w_j(\nu)(\alpha+\delta)\quad\text{ for all }j=1,\dots,m \label{eq:algo3UneqUpper}.
\end{align}
We fix all components of $\nu$ except one, say $\nu_i$. The argumentation is similar to the one above for the lower approximations. Due to the monotonicity properties of the p-value and weight, a potential $\nu_i^{\text{new}}$ fulfilling (\ref{eq:HommelAlgoStepUpper})  must be smaller than or equal to $\nu_i$. The continuity of the p-value and the weight as well as the limiting properties of the p-value yield the existence of a unique $\nu_i^{\text{new}}\leq \nu_i$ fulfilling (\ref{eq:HommelAlgoStepUpper}) .
This implies $\nu^{\text{new}}\leq\nu$. 
The identity $\nu_j^{\text{new}}=\nu_j$ for all $j\neq i$ implies $p_j(\nu_j^{\text{new}})=p_j(\nu_j)$ and, because each weight $w_j$ is non-decreasing in the $i$-th component, we obtain $w_j(\nu)(\alpha+\delta) \geq w_j(\nu^{\text{new}})(\alpha+\delta)$
for all $j\not=i$. Thus, for the updated $\nu\rightarrow\nu^{\text{new}}$, in particular, after decreasing $\delta$, we obtain for all components ($i$ included) that (\ref{eq:algo3UneqUpper}) is fulfilled. Because the start vector $\nu_0$ fulfils (\ref{eq:algo3UneqUpper}) it follows inductively that the sequence of $\nu$ is non-increasing.

It remains to argue that in each step of the algorithm, we have $\nu\in C$ implying $\nu\geq L$. Since in each step, inequality  (\ref{eq:algo3UneqUpper}) is satisfied for all components $j$, we have $\delta>0$ and the weights are positive, we obtain $p_j(\nu_j)/w_j(\nu)>\alpha$ for all $1\leq j\leq m$,  implying 

$p(\nu)=\min_{j=1}^m p_j(\nu_j)/w_j(\nu)>\alpha$. Thus, $\nu\in C$ implying $\nu\geq L$, which proves the statement.
%\end{proof}

\subsection*{Proof of (b) in Theorem \ref{theoremConvergencHommelInDetail}}
%\begin{proof}
The convergence of the sequences follows from their monotonicity and boundedness shown in part (a) of Theorem \ref{theoremConvergencHommelInDetail}.  
%The sequence of lower approximations $\lambda$ is non-decreasing component-wise and bounded from above by $p_i^{-1}(\alpha)<\infty$ for $1\leq i\leq m$. The sequence of upper approximations $\nu$ is non-increasing component-wise and bounded from below by $L$. This yields the convergence. 
We denote the limit points by $\lambda_{\infty}$ and $\nu_{\infty}$ and show that they meet equation (\ref{uniquePropertyHommel}). Note that the convergence against $\lambda_{\infty}$ and $\nu_{\infty}$ also applies for each subsequence of $\lambda$ and $\nu$.
%Given the convergence, it suffices to show that for each $\lambda_i$ and $\nu_i$ a subsequence converges to $\lambda_{\infty,i}$ and $\nu_{\infty,i}$, respectively. 
We consider for each component $i=1,\dots,m$ only the steps in which equation (\ref{eq:HommelAlgoStepLower}) is satisfied, i.e.\ the steps 
$k_{\ell}=i+\ell\cdot m$, $\ell\in\{0,1,2,\ldots\}$, of the algorithm. Given the convergence of the subsequence as well as the continuity of the p-values and weights, and the fact that $\delta$ of each step converge to $0$, we obtain that  (\ref{uniquePropertyHommel}) is fulfilled for $\lambda_{\infty}$ and $\nu_{\infty}$ in components $i=1,\dots,m$.

We finally note that by \eqref{eq:algo3Uneq}, for each component $\lambda_i$ of the sequence $\lambda$, constructed by Algorithm 3, we obtain $p_i(\lambda_i)\le \alpha$ and therefore $\lambda_i\le p^{-1}_i(\alpha)$. Hence, the convergence of this sequence to $\lambda_\infty$ that satisfies 
\eqref{uniquePropertyHommel} is true also without the assumption that $\lambda_0\le L$ for the stating point (which is required to deduce $\lambda\le L$ for all parameter points of the sequence).
%\end{proof}

\subsection*{Proof of (c) in Theorem \ref{theoremConvergencHommelInDetail}}
%\begin{proof}
    Let $\vartheta\in\mathbb{R}^m$ such that property (\ref{uniquePropertyHommel}), is met. By summing up both sides of this equation over all $1\leq i\leq m$, we obtain:
\begin{align}
    \sum_{i=1}^m p_i(\vartheta_i)=\alpha\cdot\sum_{i=1}^mw_i(\vartheta)=\alpha.\label{forProofHommel1}
\end{align}
Assume now that there exists another $\vartheta'\in\mathbb{R}^m$ with $\vartheta'\neq\vartheta$ such that property (\ref{uniquePropertyHommel}) and thus (\ref{forProofHommel1}) are met and show that this assumption leads to a contradiction. 

To this end we define $\tau=(\tau_1,\dots,\tau_m)$ by 
$\tau_i:=\max\{\vartheta_i, \vartheta'_i\}$. We consider an arbitrary component $i$ and assume w.l.o.g.\ that $\tau_i=\vartheta_i$. Then $w_i(\vartheta)\leq w_i(\tau)$ because $\vartheta_j\leq\tau_j$ 
and $w_i$ is non-decreasing in all components $j\neq i$. With this and our assumption that $\vartheta$ meets property (\ref{uniquePropertyHommel}) we obtain
\begin{align}
    p_i(\tau_i)=p_i(\vartheta_i)=\alpha\cdot w_i(\vartheta)\leq\alpha\cdot w_i(\tau)~.\label{forProofHommel2}
\end{align}
The same argument applies to all $\tau_j$, $j=1,\ldots,m$, with $\vartheta'_j$ instead of $\vartheta_j$ if 
$\tau_j=\vartheta'_j$. Hence, $\tau$ meets the first condition in \eqref{eq:initA3} of a starting point for Algorithm~\ref{algorithm:Hommel}. Since the proof of the convergence properties of the sequence $\lambda$ for parts (a) and (b) of Theorem~\ref{theoremConvergencHommelInDetail} do not require the property $\lambda_0\leq L$, but only the property that $p_i(\lambda_{i})\leq \alpha\cdot w_i(\lambda)$ for all $1\leq i\leq m$, we can conclude convergence of the sequence $\lambda$ which is obtained by Algorithm~\ref{algorithm:Hommel} with starting point $\lambda_{0}:=\tau$, namely to some $\tau_{\infty}\in\mathbb{R}^m$ with $\tau_{\infty}\geq\lambda_0=\tau$. From (a) and (b) of Theorem~\ref{theoremConvergencHommelInDetail}, we get that $\tau_{\infty}$ meets (\ref{uniquePropertyHommel}) and thus $\sum_{i=1}^m p_i(\tau_{i,\infty})=\alpha$.

Because $\vartheta\neq\vartheta'$ and $\vartheta, \vartheta'\leq\tau$ there exists at least one index $j$ such that $\vartheta_j<\tau_j$ or $\vartheta'_j<\tau_j$, w.l.o.g.\ $\vartheta_j<\tau_j$. Since $\lambda\ge \tau$ for all points $\lambda$ of the
non-decreasing sequence obtained from Algorithm 3 with starting point $\tau$, we obtain $\vartheta\le \tau\le \tau_\infty$ and $\vartheta_j< \tau_j\le \tau_{j,\infty}$. Because the local p-values are assumed to be strictly increasing in their parameter, 
we end up with the following contradictory statement:
\begin{align*}
    \alpha=\sum_{i=1}^mp_i(\vartheta_i) < \sum_{i=1}^mp_i(\tau_{i,\infty})=\alpha,
\end{align*}
whereby the first identity follows from \eqref{forProofHommel1} and the second one from statement (b) of Theorem~\ref{theoremConvergencHommelInDetail}.

In summary, there can exist at most one $\vartheta\in\mathbb{R}^m$ such that property (\ref{uniquePropertyHommel}) is fulfilled for the given $\alpha$.
%\end{proof}

\subsection*{Proof of (d) in Theorem \ref{theoremConvergencHommelInDetail}}
%\begin{proof}
The statement follows directly from (a) to (c) of Theorem~\ref{theoremConvergencHommelInDetail} together with the assumption of existing starting vectors $\lambda_0\leq L$ and $L\leq\nu_0\in C$.
%\end{proof}

\end{document}